\documentclass[10pt]{article}
\usepackage{times} 
\usepackage{amsmath,amssymb,amsthm,url,color}
\usepackage{enumerate}
\usepackage{graphicx}
\usepackage{epstopdf}
\def\today{\number\day\ \ifcase\month\or
 January\or February\or March\or April\or May\or June\or
 July\or August\or September\or October\or November\or December\fi
 \space \number\year}
\def\gobble#1#2{}
\def\shortdate{\number\day/\number\month/\expandafter\gobble\number\year}

\def\N{\mathbb{N}}
\def\R{\mathbb{R}}
\def\p{Pain\-lev\'e}

\def\peq{\p\ equation}
\def\peqs{\p\ equations}

\def\det{\operatorname{det}}

\def\ds{\displaystyle}

\def\Ai{\operatorname{Ai}}
\def\Bi{\operatorname{Bi}}

\newcommand{\HyperpFq}[2]{{}_{#1}F_{#2}}

\def\d{{\rm d}}\def\e{{\rm e}}
\def\a{\alpha}
\def\b{\beta}
\def\ga{\gamma}
\def\la{\lambda}
\def\k{\kappa}
\def\w{\omega}

\def\ph{\varphi}
\def\vph{\varphi}
\def\A{\mathcal{A}}
\def\B{\mathcal{B}}
\def\O{\mathcal{O}}
\def\W{\mathcal{W}}

\newcommand{\Integer}{\mathbb{Z}}

\def\Z{\Integer}

\newcommand{\deriv}[3][]{\frac{\d^{#1}{#2}}{{\d{#3}}^{#1}}}

\newtheorem{theorem}{Theorem}[section]

\newtheorem{lemma}[theorem]{Lemma}
\newtheorem{corollary}[theorem]{Corollary}
\theoremstyle{definition}

\newtheorem{remark}[theorem]{Remark}
\newtheorem{remarks}[theorem]{Remarks}
\newtheorem{conjecture}[theorem]{Conjecture}
\numberwithin{figure}{section}
\numberwithin{equation}{section}
\numberwithin{table}{section}

\newcommand{\comment}[1]{}

\def\la{\lambda}

\def\imp{\int_{-\infty}^{\infty}}

\def\beq{\begin{equation}}
\def\eeq{\end{equation}}
\def\ds{\displaystyle}

 \definecolor{dkg}{rgb}{0,0.7,0}
 \definecolor{dkr}{rgb}{0.9,0,0}
\definecolor{dkb}{rgb}{0,0,0.7}
\definecolor{purple}{rgb}{0.5,0,0.7}

\def\blue#1{\textcolor{blue}{#1}}

\definecolor{gold}{rgb}{0.83, 0.69, 0.22}

\begin{document}

\title{A Generalised Sextic Freud Weight}

\author{Peter A. Clarkson$^{1}$ and Kerstin Jordaan$^{2}$\\[2.5pt]
$^{1}$ School of Mathematics, Statistics and Actuarial Science,\\ University of Kent, Canterbury, CT2 7FS, UK\\ {\tt P.A.Clarkson@kent.ac.uk}\\[2.5pt]
$^{2}$ Department of Decision Sciences,\\
University of South Africa, Pretoria, 0003, South Africa\\ 
{\tt jordakh@unisa.ac.za}}

\maketitle

\begin{abstract}
We discuss the recurrence coefficients of orthogonal polynomials with respect to a generalised sextic Freud weight 
\[\omega(x;t,\lambda)=|x|^{2\lambda+1}\exp\left(-x^6+tx^2\right),\qquad x\in\mathbb{R},\]
with parameters $\lambda>-1$ and $t\in\mathbb{R}$. We show that the coefficients in these recurrence
relations can be expressed in terms of Wronskians of generalised hypergeometric functions ${}_1F_2(a_1;b_1,b_2;z)$. We derive a nonlinear discrete as well as a system of differential equations satisfied by the recurrence coefficients and use these to investigate their asymptotic behaviour. We conclude by highlighting a fascinating connection between generalised quartic, sextic, octic and decic Freud weights when expressing their first moments in terms of generalised hypergeometric functions. 
\end{abstract}

\section{Introduction}
In this paper we are concerned with semi-classical polynomials which are orthogonal with respect to a symmetric weight on an unbounded interval. Orthogonal polynomials find application in various branches of mathematics such as approximation theory, special functions, continued fractions and integral equations. Examples of classical orthogonal polynomals on infinite intervals include Hermite and Laguerre polynomials. 

 Although Szeg\H{o} pioneered much of what is known on the theory of orthogonal polynomials on finite intervals, he did not carry his ideas over to infinite intervals, despite there being significant differences. It was only in the second half of the 20th century, starting with the work of G\'{e}za Freud on orthogonal polynomials on $\R$, that the study of Freud-type polynomials, and their generalisations, flourished. One of Freud's original aims was to extend the theory of best approximations using Jackson-Bernstein type estimates to the real line and, since a realistic expectation was that orthogonal expansions could serve as near-best approximations, a natural approach was to explore properties of orthogonal polynomials \cite{refNevai86,refMhaskar}.

 Freud \cite{refFreud76} studied polynomials $\{ P_n (x) \}_{n=0}^{\infty}$ orthogonal on the real line with respect to a class of exponential-type weights, known as Freud weights, given by \beq\nonumber \label{conj}w(x) = |x|^{\rho}\exp(-|x|^{m}), \qquad m\in \N,\eeq 
with $\rho>-1$, for more details see \cite{refFreud71,refNevai84a,refNevai86,refLevinLubinsky01,refMhaskar}. Freud conjectured that the asymptotic behaviour of recurrence coefficients $\b_n$ in the recurrence relation
 \beq \label{orthogonalrecur}
P_{n+1}(x)= xP_{n}(x)-\b_n P_{n-1}(x), 
\eeq with $P_{-1}(x)=0$ and $P_0(x)=1$, {is} given by
\beq \label{Freudconj}
\ds \lim_{n\rightarrow \infty} \frac{\b_n}{n^{2/m}}= \left[ \frac{\Gamma(\tfrac{1}{2}m) \Gamma(1+\tfrac{1}{2}m)}{\Gamma(m+1)}\right]^{2/m}.
\eeq 

 Freud \cite{refFreud76} showed that if the limit exists for $m\in2\Z$, then it is equal to the expression in \eqref{Freudconj} and proved existence of the limit \eqref{Freudconj} for $m=2,4,6$ using a technique that gives rise to an infinite system of nonlinear equations, called Freud equations. A general proof of Freud's conjecture was given by Lubinsky, Mhaskar and Saff \cite{refLubinskyMS}; see also \cite{refDamelin,refFreud71,refFreud76, refMagnus85,refNevai86}. Freud explored other properties, such as the asymptotic behaviour of the polynomials using the recurrence coefficients and the asymptotic behaviour of the greatest zero \cite{refFreud86}. 
For recent contributions on the asymptotic behaviour of the recurrence coefficients associated with Freud-type exponential weights and zeros of the associated polynomials see, for example
\cite{refANSVA15,refLevinLubinsky01,refLubinskyMS,
refMNT85,refMNT86,refMNZ85,refNevai84a,refRakhmanov}.

 {Iserles and Webb \cite{refIW} discuss orthogonal systems in $L^2(\R)$ which give rise to a real skew-symmetric, tridiagonal, irreducible differentiation matrix. Such systems are important since they are stable by design and, if necessary, preserve Euclidean energy for a variety of time-dependent partial differential equations. Iserles and Webb \cite{refIW} prove that there is a one-to-one correspondence between such an orthogonal system $\{\phi_n(x)\}_{n=0}^\infty$ and a sequence of polynomials $\{P_n(x)\}_{n=0}^\infty$ which are orthogonal with respect to a symmetric weight.}
 
 In this paper we consider polynomials orthogonal with respect to the \textit{generalised sextic Freud weight} 
 \beq \label{genFreud6}
\w(x;t,\la)=|x|^{2\la +1}\exp(-x^6+tx^2),\qquad x,t\in\R,\eeq 
with $\la>-1$. The weight \eqref{genFreud6} was briefly investigated by Freud  \cite{refFreud76} though lies in a more general framework that was given much earlier by Shohat \cite{refShohat39}. The special case when $\la=-\tfrac12$ and $t = 0$ was investigated by Sheen in his PhD thesis \cite{refSheenPhD} and some asymptotics are given in \cite{refSheen}. %
{Its and Kitaev \cite{refIK90}, see also \cite[\S5]{refFIK91}, investigated the weight
\beq\label{wIK} \w(x;\b,q_1,q_2) = \exp\left\{-\b(\tfrac12x^2 + q_1x^4 +q_2 x^6)\right\},\eeq
with $\b$, $q_1$ and $q_2$ parameters such that $q_1<0$ and $0\leq 5q_2<4q_1^2$, in their study of the  continuous limit for the Hermitian matrix model in connection with the nonperturbative theory of two-dimensional quantum gravity.}%
\comment{Differential systems satisfied by weights of the form $\w(x)=\exp\{-V(x)\}$, where $V(x)$ is an even polynomial with positive leading coefficient, are discussed by Bertola, Eynard and Harnad  \cite{refBEH03,refBEH06}.}

The paper is organised as follows: in \S\ref{sec:op}, we review some properties of orthogonal polynomials with symmetric weight while we prove that the first moment of the generalised sextic Freud weight is a linear combination of generalised hypergeometric functions $\HyperpFq12(a_1;b_1,b_2;z)$ and use this to derive an expression for the recurrence coefficents in terms of Wronskians of such generalised hypergeometric functions in \S\ref{sec:Freud6weight}. In \S\ref{sec:betaneqns} we derive a nonlinear difference equation satisfied by recurrence coefficients of generalised sextic Freud polynomials that is a special case of the second member of the discrete \p\ I hierarchy. We also derive a system of differential equations that the recurrence coefficients satisfy and use the difference and differential equations to study the asymptotic behaviour of the coefficients when the degree $n$ and parameter $t$ tend to infinity. 
Finally, in \S\ref{sec:Gen46810Freud} we derive the first moments of the generalised octic and decic Freud weights and show that the moments as well as the differential equations satisfied by the moments have a predictable structure as the order of the polynomial in the exponential factor of the Freud weight increases. 

\section{\label{sec:op}Orthogonal polynomials with symmetric weights}
Let $P_n(x)$, $n\in\N$, be the monic orthogonal polynomial of degree $n$ in $x$ with respect to a positive weight $\w(x)$ on 
the real line $\R$, such that
\beq \nonumber\imp P_m(x)P_n(x)\,\w(x)\,\d x = h_n\delta_{m,n},\qquad h_n>0,\label{eq:norm}\eeq
where $\delta_{m,n}$ denotes the {Kronecker} delta. 
One of the most important properties of orthogonal polynomials is that they satisfy a three-term recurrence relationship of the form
\beq \nonumber \label{eq:3trr}
P_{n+1}(x)=xP_n(x)-\a_nP_n(x)-\b_nP_{n-1}(x),
\eeq
where the coefficients $\a_n$ and $\b_n$ are given by the integrals
\[
\a_n = \frac1{h_n}\imp xP_n^2(x)\,\w(x)\,\d x,\qquad \b_n = \frac1{h_{n-1}}\imp xP_{n-1}(x)P_n(x)\,\w(x)\,\d x,
\]
with $P_{-1}(x)=0$ and $P_{0}(x)=1$. For symmetric weights, i.e.\ $\w(x)=\w(-x)$, then clearly $\a_n \equiv0$. Hence for symmetric weights, the monic orthogonal polynomials $P_n(x)$, $n\in\N$, satisfy the three-term recurrence relation
\beq\label{eq:srr}
P_{n+1}(x)=xP_n(x)-\b_nP_{n-1}(x).
\eeq
The relationship between the recurrence coefficient $\b_n$ and the normalisation constants $h_n$ is given by
\beq\label{bn:hn} h_n=\b_n h_{n-1}.\eeq
The coefficient $\b_n$ in the recurrence relation \eqref{eq:srr} can be expressed in terms of a determinant whose entries are given in terms of the moments associated with the weight $\w(x)$. Specifically
\beq\label{def:bn}
\b_n = \frac{\Delta_{n+1}\Delta_{n-1}}{\Delta_{n}^2},\eeq
where $\Delta_n$ is the Hankel determinant 
\beq\label{eq:detsDn}
\Delta_n
=\left|\begin{matrix} \mu_0 & \mu_1 & \ldots & \mu_{n-1}\\
\mu_1 & \mu_2 & \ldots & \mu_{n}\\
\vdots & \vdots & \ddots & \vdots \\
\mu_{n-1} & \mu_{n} & \ldots & \mu_{2n-2}\end{matrix}\right|,\qquad n\geq1,\eeq
with $\Delta_0=1$, $\Delta_{-1}=0$, and $\mu_k$, the $k$th moment, is given by the integral
\beq \nonumber \label{eq:moment}
\mu_k=\imp x^k\w(x)\,\d x.\eeq
For symmetric weights then clearly $\mu_{2k-1}\equiv 0$, for $k=1,2,\ldots\ $.

 For symmetric weights 
it is possible to write the Hankel determinant $\Delta_n$ in terms of the product of two Hankel determinants, as given in the following lemma. The decomposition depends on whether $n$ is even or odd.
\begin{lemma}\label{lem:Deltan}
Suppose that 
$\A_n$ and $\B_n$ are the Hankel determinants given by
\begin{align}\label{def:AnBn} \A_n &
=\left|\begin{matrix} 
\mu_0 & \mu_2 & \ldots & \mu_{2n-2}\\
\mu_2 & \mu_4 & \ldots & \mu_{2n} \\
\vdots & \vdots & \ddots & \vdots \\
 \mu_{2n-2} & \mu_{2n}& \ldots & \mu_{4n-4}
\end{matrix}\right|,\qquad
\B_n 
=\left|\begin{matrix} 
\mu_2 & \mu_4 & \ldots & \mu_{2n}\\
\mu_4 & \mu_6 & \ldots & \mu_{2n+2} \\
\vdots & \vdots & \ddots & \vdots \\
 \mu_{2n} & \mu_{2n+2}& \ldots & \mu_{4n-2}
\end{matrix}\right|.\end{align}
Then the determinant $\Delta_n$ \eqref{eq:detsDn} is given by
\beq \Delta_{2n}=\A_n\B_n,\qquad \Delta_{2n+1}=\A_{n+1}\B_n.\label{res:lemma21}\eeq
\end{lemma}

\begin{proof}{The result is easily obtained by matrix manipulation interchanging rows and columns.}
\end{proof}

\begin{corollary} 
For a symmetric weight, the recurrence coefficient $\b_n$ is given by
\beq \b_{2n} = \frac{\A_{n+1}\B_{n-1}}{\A_n\B_n},\qquad
\b_{2n+1}= \frac{\A_{n}\B_{n+1}}{\A_{n+1}\B_n},\label{def:betan}
\eeq
where $\A_n$ and $\B_n$ are the Hankel determinants given by \eqref{def:AnBn}, with $\A_0=\B_0=1$.
\end{corollary}
\begin{proof}
{Substituting  \eqref{res:lemma21} into \eqref{def:bn} gives the result.} \end{proof}%
{\begin{remark}{\rm The expression of the Hankel determinant $\Delta_n$ for symmetric weights as a product of two determinants is  given in \cite{refCHL,refLCF}.
}\end{remark}}%

\begin{lemma}Suppose that $\w_0(x)$ is a symmetric positive weight on the real line for which all the moments exist and $\w(x;t)=\exp(tx^2)\,\w_0(x)$, with $t\in\R$, is a weight such that all the moments also exist. Then the Hankel determinants $\A_n$ and $\B_n$ given by \eqref{def:AnBn} can be written in terms of Wronskians, as follows
\beq \A_n =W\left(\mu_0,\deriv{\mu_0}{t},\ldots,\deriv[n-1]{\mu_0}{t} \right), 
\qquad \B_n =\W\left(\deriv{\mu_0}{t},\deriv[2]{\mu_0}{t},\ldots,\deriv[n]{\mu_0}{t} \right),
\label{def:AnBnW}\eeq
where
\[ \mu_0(t;\la) = \imp \exp(tx^2)\,\w_0(x)\,\d x,\]
{and $\W(\ph_1,\ph_2,\ldots,\ph_n)$ is the Wronskian given by
$$\W(\ph_1,\ph_2,\ldots,\ph_n)=\left|\begin{matrix} 
\ph_1 & \ph_2 & \ldots & \ph_n\\
\ph_1^{(1)} & \ph_2^{(1)} & \ldots & \ph_n^{(1)}\\
\vdots & \vdots & \ddots & \vdots \\
\ph_1^{(n-1)} & \ph_2^{(n-1)} & \ldots & \ph_n^{(n-1)}
\end{matrix}\right|,\qquad \ph_j^{(k)}=\deriv[k]{\ph_j}{t}.$$}
\end{lemma}
\begin{proof} If $\w(x;t)=\exp(tx^2)\,\w_0(x)$, with $t\in\R$ then
\[ \mu_{2n}= \imp x^{2n} \exp(tx^2)\,\w_0(x)\,\d x = \deriv[n]{}{t} \imp \exp(tx^2)\,\w_0(x)\,\d x=\deriv[n]{\mu_0}{t},\qquad n=1,2,\ldots\ , \]
and so it follows from \eqref{def:AnBn} that $\A_n$ and $\B_n$ are given by \eqref{def:AnBnW}.
\end{proof}

\begin{lemma} If $\A_n$ and $\B_n$ are Wronskians given by \eqref{def:AnBnW}, with $\A_0=\B_0=1$,
then
 \begin{align}\label{eq:dodgson} &\A_n\deriv{\B_n}{t}-\B_n\deriv{\A_n}{t}=\A_{n+1}\B_{n-1},\qquad
\B_{n}\deriv{\A_{n+1}}{t}-\A_{n+1}\deriv{\B_{n}}{t}=\A_{n+1}\B_{n}.\end{align} 
\end{lemma}
\begin{proof}
See, for example, Vein and Dale \cite[\S6.5.1]{refVeinDale}.
\end{proof}

\begin{corollary}For a symmetric weight, if $\A_n$ and $\B_n$ are Wronskians given by \eqref{def:AnBnW} then the recurrence coefficient $\b_n$ is given by
\beq \nonumber \b_{2n} = \deriv{}{t} \ln \frac{\B_n}{\A_n},\qquad
\b_{2n+1}= \deriv{}{t} \ln\frac{\A_{n+1}}{\B_n}. \label{def:betant}
\eeq
\end{corollary}
\begin{proof} {This follows from \eqref{def:betan} and \eqref{eq:dodgson}.}\end{proof}

\begin{lemma} Let $\w_0(x)$ be a symmetric positive weight on the real line for which all the moments exist and let $\w(x;t)=\exp(tx^2)\,\w_0(x)$, with $t\in\R$, is a weight such that all the moments of exist. Then the recurrence coefficient $\b_{n}(t)$ satisfies the Volterra, or the Langmuir lattice, equation
\beq \deriv{\b_{n}}{t} = \b_{n}(\b_{n+1}-\b_{n-1}).\label{eq:langlat}\eeq
\end{lemma}
 We remark that the differential-difference equation \eqref{eq:langlat} is also known as the discrete KdV equation, or the Kac-van Moerbeke lattice \cite{refKvM75}. 

\begin{proof} {See, for example, Van Assche \cite[Theorem 2.4]{refWVAbk} and  Wang, Zhu and Chen \cite{refWZC}.}
\end{proof}

 The weights of classical orthogonal polynomials satisfy a first-order ordinary differential equation, the \textit{Pearson equation}
\begin{equation}\label{eq:Pearson}
\deriv{}{x}[\sigma(x)\w(x)]=\tau(x)\w(x),
\end{equation}
where $\sigma(x)$ {is a monic polynomial} of degree at most $2$ and $\tau(x)$ is a polynomial with degree $1$. However for \textit{semi-classical} orthogonal polynomials, the weight function $\w(x)$ satisfies the Pearson equation (\ref{eq:Pearson}) with either deg$(\sigma)>2$ or deg$(\tau)\neq 1$ (cf. \cite{refHvR,refMaroni}). 
%

For example, the generalised sextic Freud weight \eqref{genFreud6}
satisfies the Pearson equation \eqref{eq:Pearson} with \[\sigma(x)=x,\qquad\tau(x)=2\la+2+2tx^2-6x^6.\]
For further information about orthogonal polynomials see, for example \cite{refChihara78,refIsmail,refSzego}.

\section{Generalised sextic Freud weight}\label{sec:Freud6weight}

In this section we are concerned with the generalised sextic Freud weight
\beq \label{freud6g}
\w(x;t,\la)=|x|^{2\la+1}\exp\left(-x^6+tx^2\right),\qquad \la>-1,\qquad x,t\in\R.\eeq
\begin{lemma}\label{lem:Freud6weight}
For the generalised sextic Freud weight \eqref{freud6g},
the first moment is given by
\begin{align}
\mu_0(t;\la)& =\int_{-\infty}^{\infty} |x|^{2\la+1}\exp(-x^6+tx^2)\,\d x = \int_0^\infty s^{\la}\exp(ts-s^3)\,\d s \nonumber\\
& = \tfrac13\Gamma(\tfrac13\la+\tfrac13) \;\HyperpFq12(\tfrac13\la+\tfrac13;\tfrac13,\tfrac23;(\tfrac13t)^3) 
+ \tfrac13 \,t\,\Gamma(\tfrac13\la+\tfrac23) \;\HyperpFq12(\tfrac13\la+\tfrac23;\tfrac23,\tfrac43;(\tfrac13t)^3)\nonumber\\ &\qquad\qquad
+ \tfrac16\,t^2\,\Gamma(\tfrac13\la+1) \;\HyperpFq12(\tfrac13\la+1;\tfrac43,\tfrac53;(\tfrac13t)^3), 
\label{eq:mu0}\end{align}
where 
$\HyperpFq12(a_1;b_1,b_2;z)$ is the generalised hypergeometric function.
\end{lemma}
\begin{proof}
First we shall show that
\[\mu_0(t;\la) = \int_0^\infty s^{\la}\exp(ts-s^3)\,\d s ,\]
is a solution of the third order equation
\beq \deriv[3]{\vph}{t}-\tfrac13t\deriv{\vph}{t}-\tfrac13(\la+1)\vph=0.\label{eq6}\eeq
Following Muldoon \cite{refMul77}, if we seek a solution of \eqref{eq6} in the form
\[ \vph(t)=\int_0^\infty \e^{st} \,v(s)\,\d s,\]
then
\[ \begin{split}
\deriv[3]{\vph}{t}-\tfrac13t\deriv{\vph}{t}-\tfrac13(\la+1)\vph&=\int_0^\infty \e^{st} \left\{ s^3v(s) -\tfrac13 tsv(s)-\tfrac13(\la+1)v(s)\right\}\,\d s\\
&=\int_0^\infty \e^{st} \left\{ s^3v(s) +\tfrac13 v(s)+\tfrac13 s\deriv{v}{s}-\tfrac13(\la+1)v(s)\right\}\,\d s=0,
\end{split}\] 
using integration by parts and assuming that $\lim_{s\to\infty} sv(s)\e^{st}=0$.
Therefore for $\vph(t)$ to be a solution of \eqref{eq6} then $v(s)$ necessarily satisfies
\[ s\deriv{v}{s}+(3s^3-\la)v=0,\] and so \[v(s) = s^{\la}\exp(-s^3).\] Hence $\mu_0(t;\la)$ satisfies \eqref{eq6}.
The general solution of equation \eqref{eq6} is given by
\[\begin{split} \vph(t)&=c_1 \;\HyperpFq12(\tfrac13\la+\tfrac13;\tfrac13,\tfrac23;(\tfrac13t)^3) + c_2t\;\HyperpFq12(\tfrac13\la+\tfrac23;\tfrac23,\tfrac43;(\tfrac13t)^3) 
+ c_3t^2 \;\HyperpFq12(\tfrac13\la+1;\tfrac43,\tfrac53;(\tfrac13t)^3),\end{split}\]
with $c_1$, $c_2$ and $c_3$ arbitrary constants. This can be derived from the third order equation satisfied by $\HyperpFq12(a_1;b_1,b_2;z)$ given in \S16.8(ii) of the DLMF \cite{refNIST}, i.e.
\beq z^2\deriv[3]{w}{z} + z(b_1+b_2+1)\deriv[2]{w}{z}+(b_1b_2 -z)\deriv{w}{z} - a_1w=0,\label{eq:1F2}\eeq
which has general solution
\begin{align*} w(z)= c_1 &\HyperpFq12(a_1;b_1,b_2;z) + c_2 z^{1-b_1}\HyperpFq12(1+a_1-b_1;2-b_1,1-b_1+b_2;z)\\ &+c_3z^{1-b_2} \HyperpFq12(1+a_1-b_2;1+b_1-b_2,2-b_2;z),
\end{align*}
with $c_1$, $c_2$ and $c_3$ constants. Note that making the transformation $w(z)=\vph(t)$, with $z=(\tfrac13t)^3$, in \eqref{eq:1F2} gives
\beq \nonumber t^2\deriv[3]{\vph}{t}+3t(b_1+b_2-1)\deriv[2]{\vph}{t}+\left[(3b_1-2)(3b_2-2)-\tfrac13t^3\right]\deriv{\vph}{t}-a_1t^2\vph=0.\eeq
Consequently setting $a_1=\tfrac13(\la+1)$, $b_1=\tfrac13$ and $b_2=\tfrac23$ we have
\[ \begin{split}\mu_0(t;\la) &= \int_0^\infty s^{\la}\exp(ts-s^3)\,\d s \\
&=c_1 \;\HyperpFq12(\tfrac13\la+\tfrac13;\tfrac13,\tfrac23;(\tfrac13t)^3) + c_2t\;\HyperpFq12(\tfrac13\la+\tfrac23;\tfrac23,\tfrac43;(\tfrac13t)^3) 
+ c_3t^2 \;\HyperpFq12(\tfrac13\la+1;\tfrac43,\tfrac53;(\tfrac13t)^3),\end{split}\]
where $c_1$, $c_2$ and $c_3$ are constants to be determined. Since $\HyperpFq12(a_1;b_1,b_2;0)=1$ and 
\[ \mu_0(0;\la)=\int_0^\infty s^{\la}\exp(-s^3)\,\d s = \tfrac13\Gamma(\tfrac13\la+\tfrac13),
\quad \deriv{\mu_0}{t}(0;\la)=\tfrac13\Gamma(\tfrac13\la+\tfrac23),\quad \deriv[2]{\mu_0}{t}(0;\la)=\tfrac13\Gamma(\tfrac13\la+1)
\]
then it follows that
\[ c_1=\tfrac13\Gamma(\tfrac13\la+\tfrac13),\qquad c_2=\tfrac13\Gamma(\tfrac13\la+\tfrac23),\qquad c_3=\tfrac16\Gamma(\tfrac13\la+1),\]
which gives \eqref{eq:mu0}, as required.
\end{proof}

\begin{remarks}{\label{rmks32}\rm 
\begin{enumerate}[(i)]\item[]
\item If $\la=-\tfrac12$ then
\[ \mu_0(t;-\tfrac12)=\int_{-\infty}^{\infty} \exp(-x^6+tx^2)\,\d x = \pi^{3/2}12^{-1/6}[\Ai^2(\tau)+\Bi^2(\tau)],\qquad \tau = 12^{-1/3}t,\]
where $\Ai(\tau)$ and $\Bi(\tau)$ are the Airy functions. This result is equation 9.11.4 in the DLMF \cite{refNIST}, which is due to Muldoon \cite[p32]{refMul77}, see also \cite{refReid95}.

\item The generalised sextic Freud weight \eqref{freud6g} is an example of a semi-classical weight for which the first moment $\mu_0(t;\la)$ satisfies a \textit{third order equation}. In our earlier studies of semi-classical weights \cite{refCJ14,refCJ18,refCJK}, the first moment has satisfied a second order equation. For example, for the quartic Freud weight 
\beq \label{genFreud4}
\w(x;t)=|x|^{2\la +1}\exp(-x^4+tx^2),\qquad x,t\in\R,
\eeq the first moment is expressed in terms of parabolic cylinder functions $D_{\nu}(z)$, or equivalently in terms of the confluent 
hypergeometric function $\HyperpFq11(a;b;z)$, see \cite{refCJ14,refCJ18,refCJK}.
These are classical special functions that satisfy second order equations.

\item Equation \eqref{eq6} arises in association with threefold symmetric Hahn-classical multiple orthogonal polynomials \cite{refLVA} and in connection with Yablonskii--Vorob'ev polynomials associated with rational solutions of the second \p\ equation \cite{refCM03}.
\end{enumerate}}\end{remarks} 
 The higher 
moment $\mu_k(t;\la)$ is given by
\beq \nonumber \mu_k(t;\la) = \int_{-\infty}^{\infty} x^k |x|^{2\la+1}\exp(-x^6+tx^2)\,\d x,\qquad k=0,1,2,\ldots\ ,\eeq
and so
\beq 
\mu_{2k}(t;\la) = \deriv[k]{}{t} \mu_0(t;\la),\qquad \mu_{2k+1}(t;\la)=0,\eeq
with $\mu_0(t;\la)$ given by \eqref{eq:mu0}.

\begin{lemma}Suppose that $\Delta_n(t;\la)$ is the Hankel determinant given by 
\beq \nonumber \Delta_n(t;\la) =\det\big[\mu_{j+k}(t;\la)\big]_{j,k=0}^{n-1},\label{def:Delta}\eeq
and $\A_n(t;\la)$ and $\B_{n}(t;\la)$ are the Hankel determinants given by
\beq
\A_n(t;\la) =\det\big[\mu_{2j+2k} (t;\la)\big]_{j,k=0}^{n-1},\qquad 
\B_{n}(t;\la) =\det\big[\mu_{2j+2k+2}(t;\la) \big]_{j,k=0}^{n-1},\label{def:AnBn2}\eeq
then
\beq \nonumber\Delta_{2n}(t;\la)=
\A_n(t;\la)\B_{n}(t;\la),\qquad \Delta_{2n+1}(t;\la)=
\A_{n+1}(t;\la)\B_{n}(t;\la).\eeq
\end{lemma}

\begin{lemma}
If $\A_n(t;\la)$ and $\B_{n}(t;\la)$ are given by \eqref{def:AnBn2} then $\B_{n}(t;\la) = \A_n(t;\la+1)$.
\end{lemma}
\begin{proof}
Since
\[ \mu_{2k+2}(t;\la) = \int_0^\infty s^{\la+k+1}\exp(ts-s^3)\,\d s = \mu_{2k}(t;\la+1),\]
then the result immediately follows.
\end{proof}

\begin{lemma} For the generalised sextic Freud weight \eqref{freud6g}, the associated monic polynomials $P_n(x)$ satisfy the recurrence relation 
\beq \label{eq:3rr} P_{n+1}(x) = xP_n(x) - \b_{n}(t;\la) P_{n-1}(x),\qquad n=0,1,2,\ldots\ ,\eeq
with $P_{-1}(x)=0$ and $P_0(x)=1$, where
\begin{subequations}\begin{align*} \b_{2n}(t;\la) &
= \frac{\A_{n+1}(t;\la)\A_{n-1}(t;\la+1)}{\A_n(t;\la)\A_{n}(t;\la+1)}=\deriv{}{t}\ln \frac{\A_{n}(t;\la+1)}{\A_n(t;\la)},\\
\b_{2n+1}(t;\la)&
= \frac{\A_{n}(t;\la)\A_{n+1}(t;\la+1)}{\A_{n+1}(t;\la)\A_{n}(t;\la+1)}=\deriv{}{t}\ln \frac{\A_{n+1}(t;\la)}{\A_{n}(t;\la+1)}.
\end{align*}\end{subequations}
where $\A_n(t;\la)$ is the Wronskian given by
\[\A_n(t;\la)=\W\left(\mu_0,\deriv{\mu_0}{t},\ldots,\deriv[n-1]{\mu_0}{t} \right),\]
with 
\begin{align*}
\mu_0(t;\la)& =\int_{-\infty}^{\infty} |x|^{2\la+1}\exp\left(-x^6+tx^2\right)\,\d x \\
& = \tfrac13\Gamma(\tfrac13\la+\tfrac13) \;\HyperpFq12\left(\tfrac13\la+\tfrac13;\tfrac13,\tfrac23;(\tfrac13t)^3\right) 
+ \tfrac13 \,t\,\Gamma(\tfrac13\la+\tfrac23) \;\HyperpFq12\left(\tfrac13\la+\tfrac23;\tfrac23,\tfrac43;(\tfrac13t)^3\right)\\ &\qquad\qquad
+ \tfrac16\,t^2\,\Gamma(\tfrac13\la+1) \;\HyperpFq12\left(\tfrac13\la+1;\tfrac43,\tfrac53;(\tfrac13t)^3\right). 
\end{align*}
\end{lemma}

\section{Equations satisfied by the recurrence coefficient}\label{sec:betaneqns}
In this section we derive a discrete equation and a differential equation satisfied by recurrence coefficient $\b_{n}(t;\la)$ and discuss this asymptotics of $\b_{n}(t;\la)$ as $n\to\infty$ and $t\to\pm\infty$.

\begin{lemma} The recurrence coefficient $\b_{n}(t;\la)$ satisfies the nonlinear discrete equation
\begin{align}6\b_{n} \big(\b_{n+2} \b_{n+1} &+ \b_{n+1}^2 + 2 \b_{n+1} \b_{n} + \b_{n+1} \b_{n-1} + \b_{n}^2 + 2 \b_{n}\b_{n-1} + \b_{n-1}^2 + \b_{n-1} \b_{n-2}\big)\nonumber\\
& -2t\b_{n}= n+\ga_n,\label{eq:dPI2}\end{align}
with $\ga_n=(\la+\tfrac12)[1-(-1)^n]$.
\end{lemma}

\begin{proof}{The fourth order nonlinear discrete equation \eqref{eq:dPI2} when $t=0$ was derived by Freud \cite{refFreud76};
see also Van Assche \cite[\S2.3]{refVanAssche07}. It is straightforward to modify the proof for the case when $t\not=0$; see also Wang, Zhu and Chen \cite{refWZC}.}%
\end{proof}

\begin{remark}{\label{rmk42}\rm
The discrete equation \eqref{eq:dPI2} is a special case of dP$_{\rm I}^{(2)}$, the second member of the discrete \p\ I hierarchy
which is given by 
\begin{align}c_4\b_{n} \big(\b_{n+2} \b_{n+1} &+ \b_{n+1}^2 + 2 \b_{n+1} \b_{n} + \b_{n+1} \b_{n-1} + \b_{n}^2 + 2 \b_{n}\b_{n-1} + \b_{n-1}^2 + \b_{n-1} \b_{n-2}\big)\nonumber\\
& +c_3\b_{n} \big(\b_{n+1}+\b_n+\b_{n-1})+ c_2\b_{n}= c_1+c_0(-1)^n+n,\label{eq:gendPI2}\end{align}
with $c_j$, $j=0,1,\ldots,4$ constants.
Cresswell and Joshi \cite{refCJ99a} show that:
\begin{itemize}
\item if $c_0=0$, then the continuum limit of \eqref{eq:gendPI2} is equivalent to
\beq \deriv[4]{w}{z}=10w\deriv[2]{w}{z}+5\left(\deriv{w}{z}\right)^{2} -10w^3+z,\label{eq:PI2}\eeq
which is P$_{\rm I}^{(2)}$, the second member of the first \p\ hierarchy \cite{refKud97}; 
\item if $c_0\not=0$, then the continuum limit of \eqref{eq:gendPI2} 
 is equivalent to
\beq \nonumber \deriv[4]{w}{z}=10w^2\deriv[2]{w}{z}+10w\left(\deriv{w}{z}\right)^{2} -6w^5+zw+\a, \label{eq:PII2}\eeq
where $\a$ is a constant, which is P$_{\rm II}^{(2)}$, the second member of the second \p\ hierarchy \cite{
refFN}. \end{itemize}
{We note that equation \eqref{eq:gendPI2} with $c_1=c_0=0$ is equation (8) in \cite{refIK90} and equation \eqref{eq:PI2} is given in
 \cite{refBMP,refFIK91,refIK90}. }%
This is analogous to the situation for the general discrete \p\ I equation
\beq c_3\b_{n} \big(\b_{n+1}+\b_n+\b_{n-1})+ c_2\b_{n}= c_1+c_0(-1)^n+n,\label{eq:gendPI}\eeq
with $c_j$, $j=0,1,2,3$ constants. If $c_0=0$ then the continuum limit of \eqref{eq:gendPI} is equivalent to the first \p\ equation
\[ \deriv[2]{w}{z}=6w^2+z,\]
whilst if $c_0\not=0$ then the continuum limit of \eqref{eq:gendPI} is equivalent to the second \p\ equation
\[ \deriv[2]{w}{z}=2w^3+zw+\a,\]
where $\a$ is a constant, see \cite{refCJ99a} for details.
}
\end{remark}

\begin{lemma} The recurrence coefficient $\b_{n}(t;\la)$ satisfies the system
\begin{subequations}\label{sys:bn}\begin{align}
&\deriv[2]{\b_{n}}{t}-3(\b_{n}+\b_{n+1})\deriv{\b_{n}}{t}+\b_{n}^3+6\b_{n}^2\b_{n+1}+3\b_{n}\b_{n+1}^2-\tfrac13t\b_{n}=\tfrac16(n+\ga_n),\\
&\deriv[2]{\b_{n+1}}{t}+3(\b_{n}+\b_{n+1})\deriv{\b_{n+1}}{t}+\b_{n+1}^3+6\b_{n+1}^2\b_{n}+3\b_{n+1}\b_{n}^2-\tfrac13t\b_{n+1}=\tfrac16(n+1+\ga_{n+1}),
\end{align}\end{subequations}
with $\ga_n=(\la+\tfrac12)[1-(-1)^n]$.
\end{lemma}
\begin{proof} 
Following Magnus \cite[Example 5]{refMagnus95}, from the Langmuir lattice \eqref{eq:langlat} we have
\begin{subequations}\label{sys:bn2}
\begin{align} 
 \deriv{\b_{n-1}}{t} &= \b_{n-1}(\b_{n}-\b_{n-2})\nonumber\\
 &= 
\b_{n-1}^2 + 3 \b_{n-1} \b_{n} + \b_{n-1} \b_{n+1} 
+ \b_{n}^2 + 2 \b_{n}\b_{n+1} + \b_{n+1}^2 + \b_{n+1} \b_{n+2} - \frac{n+\ga_n}{6\b_n} -\tfrac13t,\label{sys:bn2a}\\
 \deriv{\b_{n}}{t} &= \b_{n}(\b_{n+1}-\b_{n-1}),\label{sys:bn2b}\\
\deriv{\b_{n+1}}{t} &= \b_{n+1}(\b_{n+2}-\b_{n}),\label{sys:bn2c}\\
 \deriv{\b_{n+2}}{t} &= \b_{n+2}(\b_{n+3}-\b_{n+1})\nonumber\\
 &=- \b_{n-1} \b_{n} -\b_{n}^2 - 2 \b_{n} \b_{n+1} - \b_{n} \b_{n+2} 
- \b_{n+1}^2 - 3 \b_{n+1}\b_{n+2} - \b_{n+2}^2+\frac{n+1+\ga_{n+1}}{6\b_{n+1}} +\tfrac13t .\label{sys:bn2d}
 \end{align}\end{subequations}
where we have used the discrete equation \eqref{eq:dPI2} to eliminate $\b_{n+3}$ and $\b_{n-2}$.
 Solving \eqref{sys:bn2b} and \eqref{sys:bn2c} for $\b_{n+2}$ and $\b_{n-1}$ gives
 \[ \b_{n+2}=\b_n+\frac{1}{\b_{n+1}}\deriv{\b_{n+1}}{t} ,\qquad \b_{n-1}=\b_{n+1}-\frac{1}{\b_n}\deriv{\b_n}{t},\] and substitution into \eqref{sys:bn2a} and \eqref{sys:bn2d}
 yields the system \eqref{sys:bn} as required.
\end{proof}

\begin{lemma} {\label{betaasymp}The recurrence coefficient $\b_{n}(t;\la)$ has the asymptotics, as $t\to\infty$
\begin{subequations}\label{betan1}\begin{align}
\b_{2n}(t;\la)&=\frac{n}{2t}+\frac{3\sqrt{3}\,n(2n-2\la-1)}{8\,t^{5/2}}+\O(t^{-4}),\\
\b_{2n+1}(t;\la)&=\tfrac13\sqrt{3t}-\frac{4n-2\la+1}{4t}
-\frac{\sqrt{3}\,(36n^2-72\la n +12\la^2-24\la+5)}{32t^{5/2}}+\O(t^{-4}).
\end{align}\end{subequations}
and as $t\to-\infty$
\begin{subequations}\label{betan2}\begin{align}
\b_{2n}(t;\la)&=-\frac{n}{t}-\frac{3n[10n^2 + 6(2\la + 1)n + 3\la^2 + 3\la+ 2]}{t^{4}}+\O(t^{-7}),\\
\b_{2n+1}(t;\la)&=-\frac{n+\la+1}{t}-\frac{3(n+\la+1)[10n^2+(8\la+14)n+(\la+3)(\la+2)]}{t^{4}}+\O(t^{-7}).
\end{align}\end{subequations}
}\end{lemma}

\begin{proof}
First we consider $\b_1(t;\la)$ which is given by
\[ \b_1(t;\la) = \frac{\mu_2(t;\la)}{\mu_0(t;\la)} = \frac{\int_0^\infty s^{\la+1}\exp(ts-s^3)\,\d s}{\int_0^\infty s^{\la}\exp(ts-s^3)\,\d s},\]
and satisfies the equation
\beq \deriv[2]{\b_1}{t}+3\b_{1}\deriv{\b_{1}}{t}+\b_{1}^3-\tfrac13t\b_{1}=\tfrac13(\la+1).\label{eq:beta1}\eeq
Since $\mu_0(t;\la)$ given by \eqref{eq:mu0} involves the sum of three generalised hypergeometric functions then its asymptotics are not as straightforward as for a classical special function, as was the case for the generalised quartic Freud weight we discussed in \cite{refCJ18,refCJK} which involved parabolic cylinder functions.
Using Laplace's method it follows that as $t\to\infty$
\[\begin{split} \mu_0(t;\la) &= \int_0^\infty s^{\la}\exp(ts-s^3)\,\d s = t^{(\la+1)/2}\int_0^\infty \xi^{\la}\exp\{t^{3/2}\xi(1-\xi^2)\}\,\d\xi 
\\ &= {3}^{-1/4-\la/2}{t}^{\la/2-1/4}\sqrt {\pi}\exp\left(\tfrac29\sqrt{3}\,{t}^{3/2}\right){\left[1+\O(t^{-3/2})\right]} \\
\mu_2(t;\la) &= \mu_0(t;\la+1)= {3}^{-3/4-\la/2}{t}^{\la/2+1/4}\sqrt {\pi}\exp\left(\tfrac29\sqrt{3}\,{t}^{3/2}\right){\left[1+\O(t^{-3/2})\right]}
\end{split}\]
and so
\[ \b_1(t;\la) =\frac{\mu_2(t\;\la)}{\mu_0(t;\la)}= \tfrac13\sqrt{3t}{\left[1+\O(t^{-3/2})\right]} ,\qquad\text{as}\quad t\to\infty.\]
Hence we suppose that as $t\to\infty$
\[ \b_1(t;\la) = \tfrac13\sqrt{3t} + \frac{a_1}{t}+\frac{a_2}{t^{5/2}}+ \O(t^{-4}).\]
Substituting this into \eqref{eq:beta1} and equating coefficients of powers of $t$ gives 
\[ a_1= \tfrac14(2\la-1),\qquad a_2 = -{\sqrt{3}(12\la^2 - 24\la + 5)}/{32},\]
and so
\[ \b_1(t;\la) = \tfrac13\sqrt{3t} + \frac{2\la-1}{4\,t}-\frac{\sqrt{3}(12\la^2 - 24\la + 5)}{32\,t^{5/2}}+ \O(t^{-4}).\]
Also using Watson's Lemma it follows that as $t\to-\infty$
\[ \mu_0(t;\la) = {\Gamma(\la+1)}{(-t)^{-\la-1}}{\left[1+\O(t^{-3})\right]} ,\qquad \mu_2(t;\la) = {\Gamma(\la+2)}{(-t)^{-\la-2}}{\left[1+\O(t^{-3})\right]},\]
and so
\[ \b_1(t;\la)=-\frac{\la+1}{t}{\left[1+\O(t^{-3})\right]},\qquad\text{as}\quad t\to-\infty.\]
Hence we suppose that as $t\to-\infty$
\[ \b_1(t;\la) =-\frac{\la+1}{t} + \frac{b_1}{t^4} + \frac{b_2}{t^7} + \O(t^{-10}).\]
Substituting this into \eqref{eq:beta1} and equating coefficients of powers of $t$ gives 
\[ b_1=-3(\la+1)(\la+2)(\la+3),\qquad b_2= -9(\la+1)(\la+2)(\la+3)(3\la^2 + 21\la + 38).\]
Then using the Langmuir lattice \eqref{eq:langlat} it can be shown that as $t\to\infty$
\begin{align*} 
\b_2(t;\la) &= \frac{1}{2t}-\frac{3\sqrt{3}\,(2\la-1)}{8\,t^{5/2}} + \O(t^{-4}), 
&&\b_3(t;\la) = \tfrac13\sqrt{3t} + \frac{2\la-5}{4t}-\frac{\sqrt{3}\,(12\la^2-96\la+41)}{32\,t^{5/2}}+ \O(t^{-4}),\\ 
\b_4(t;\la) &= \frac{1}{t}-\frac{3\sqrt{3}\,(2\la-3)}{4\,t^{5/2}} + \O(t^{-4}), 
&&\b_5(t;\la) = \tfrac13\sqrt{3t} + \frac{2\la-9}{4t}-\frac{\sqrt{3}\,(12\la^2-168\la+149)}{32\,t^{5/2}}+ \O(t^{-4}),
\end{align*}
and as $t\to-\infty$
\begin{align*} \b_2(t;\la) &=- \frac{1}{t}-\frac{9(\la+2)(\la+3)}{t^4}+\O(t^{-7}), &&\b_3(t;\la) =- \frac{\la+2}{t}-\frac{3(\la+2)(\la+3)(\la+10)}{t^4}+\O(t^{-7}),\\ \b_4(t;\la) &=- \frac{2}{t}-\frac{18(\la+3)(\la+6)}{t^4}+\O(t^{-7}),&& \b_5(t;\la) =- \frac{\la+3}{t}-\frac{3(\la+3)(\la^2+21\la+74)}{t^4}+\O(t^{-7}).\end{align*}
From these we can see a pattern emerging for the asymptotics of $\b_{n}(t;\la)$ 
as $t\to\pm\infty$, which are different depending on whether $n$ is even or odd.

Now suppose that $u_n(t;\la)=\b_{2n}(t;\la)$ and $v_n(t;\la)=\b_{2n+1}(t;\la)$, which from \eqref{sys:bn} satisfy
\begin{subequations}\label{sys:ubn}\begin{align}
&\deriv[2]{u_n}{t}-3(u_n+v_n)\deriv{u_n}{t}+u_n^3+6u_n^2v_n+3u_nv_n^2-\tfrac13tu_n=\tfrac13n,\\
&\deriv[2]{v_n}{t}+3(u_n+v_n)\deriv{v_n}{t}+v_n^3+6v_n^2u_n+3v_nu_n^2-\tfrac13tv_n=\tfrac13(n+1+\la).
\end{align}\end{subequations}
If we suppose that as $t\to\infty$ 
\[ u_n= \frac{{\widetilde{a}}_1}{t}+\frac{{\widetilde{a}}_2}{t^{5/2}}+ \O(t^{-4}),\qquad v_n = \tfrac13\sqrt{3t} + \frac{{\widetilde{b}}_1}{t}+\frac{{\widetilde{b}}_2}{t^{5/2}}+ \O(t^{-4}),\]
with ${\widetilde{a}}_1$, ${\widetilde{a}}_2$, ${\widetilde{b}}_1$ and ${\widetilde{b}}_3$ constants,
then substituting into \eqref{sys:ubn} and equating coefficients of powers of $t$ gives \eqref{betan1}.
Also if we suppose that as $t\to-\infty$ 
\[ u_n= \frac{{\widetilde{c}}_1}{t}+\frac{{\widetilde{c}}_2}{t^{4}}+ \O(t^{-7}),\qquad v_n = \frac{{\widetilde{d}}_1}{t}+\frac{{\widetilde{d}}_2}{t^{4}}+ \O(t^{-7}),\]
with ${\widetilde{c}}_1$, ${\widetilde{c}}_2$, ${\widetilde{d}}_1$ and ${\widetilde{d}}_2$ constants,
then substituting into \eqref{sys:ubn} and equating coefficients of powers of $t$ gives \eqref{betan2}.
\end{proof}

Plots of $\b_n(t;\la)$, for $n=1,2,\ldots,10$, with $\la=-\tfrac12,\tfrac12,\tfrac32$ are given in Figure~\ref{fig:betan}. We see that there is completely different behaviour for $\b_n(t;\la)$ as $t\to\infty$, depending on whether $n$ is even or odd, which is reflected in Lemma~\ref{betaasymp}. From these plots we make the following conjecture.

{\begin{figure}[ht!]
\[\begin{array}{c@{\quad}c@{\quad}c}
\includegraphics[width=2in]{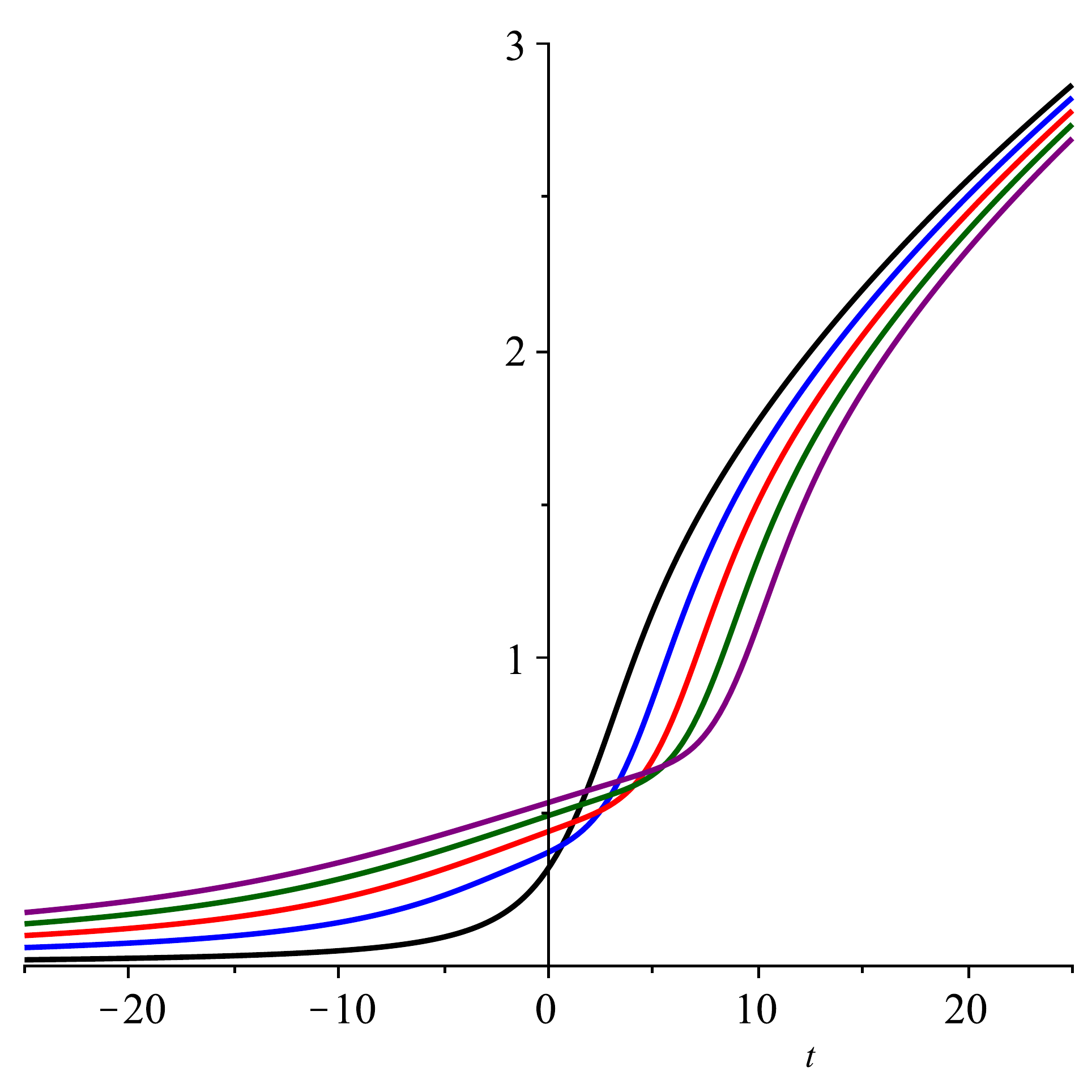} 
& \includegraphics[width=2in]{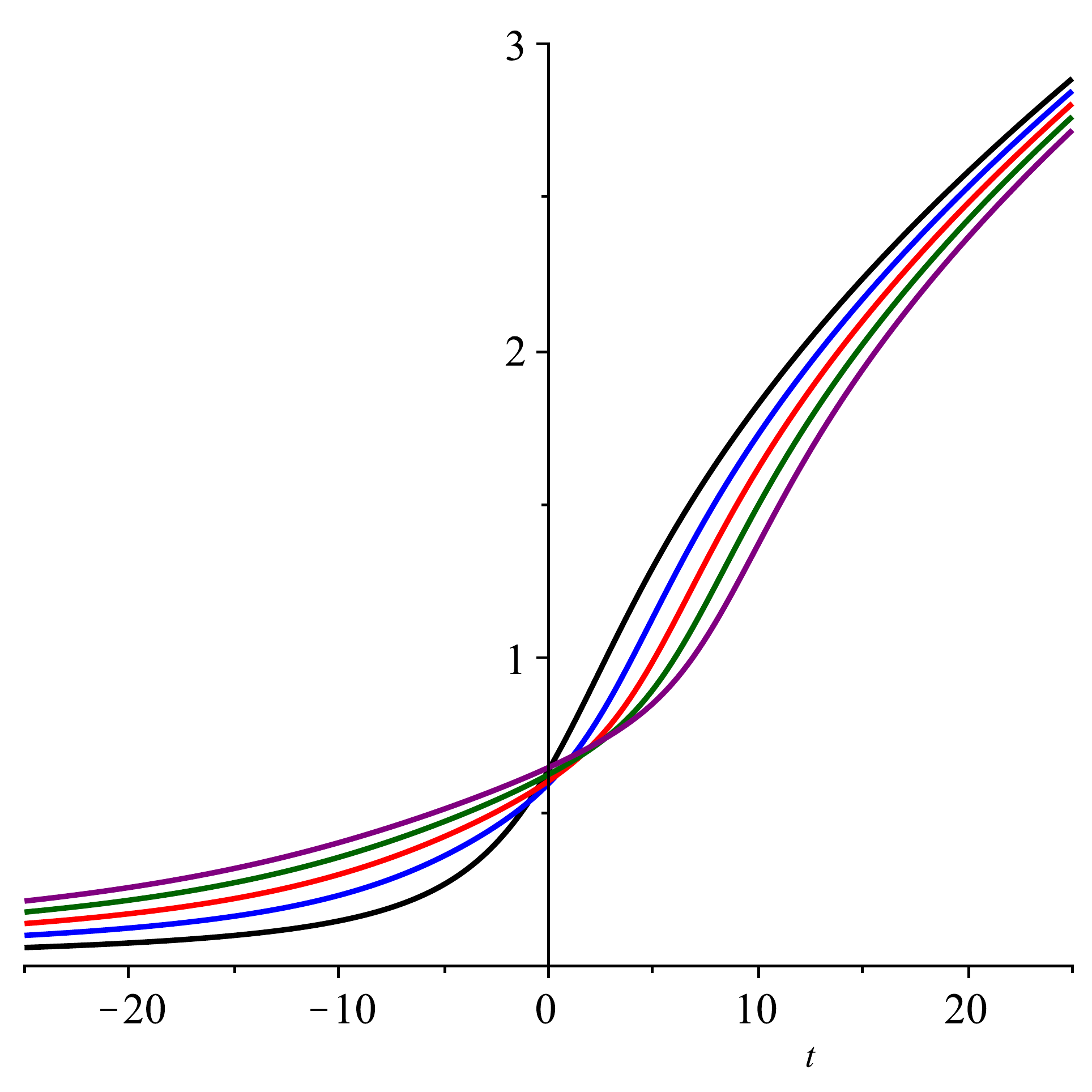} & \includegraphics[width=2in]{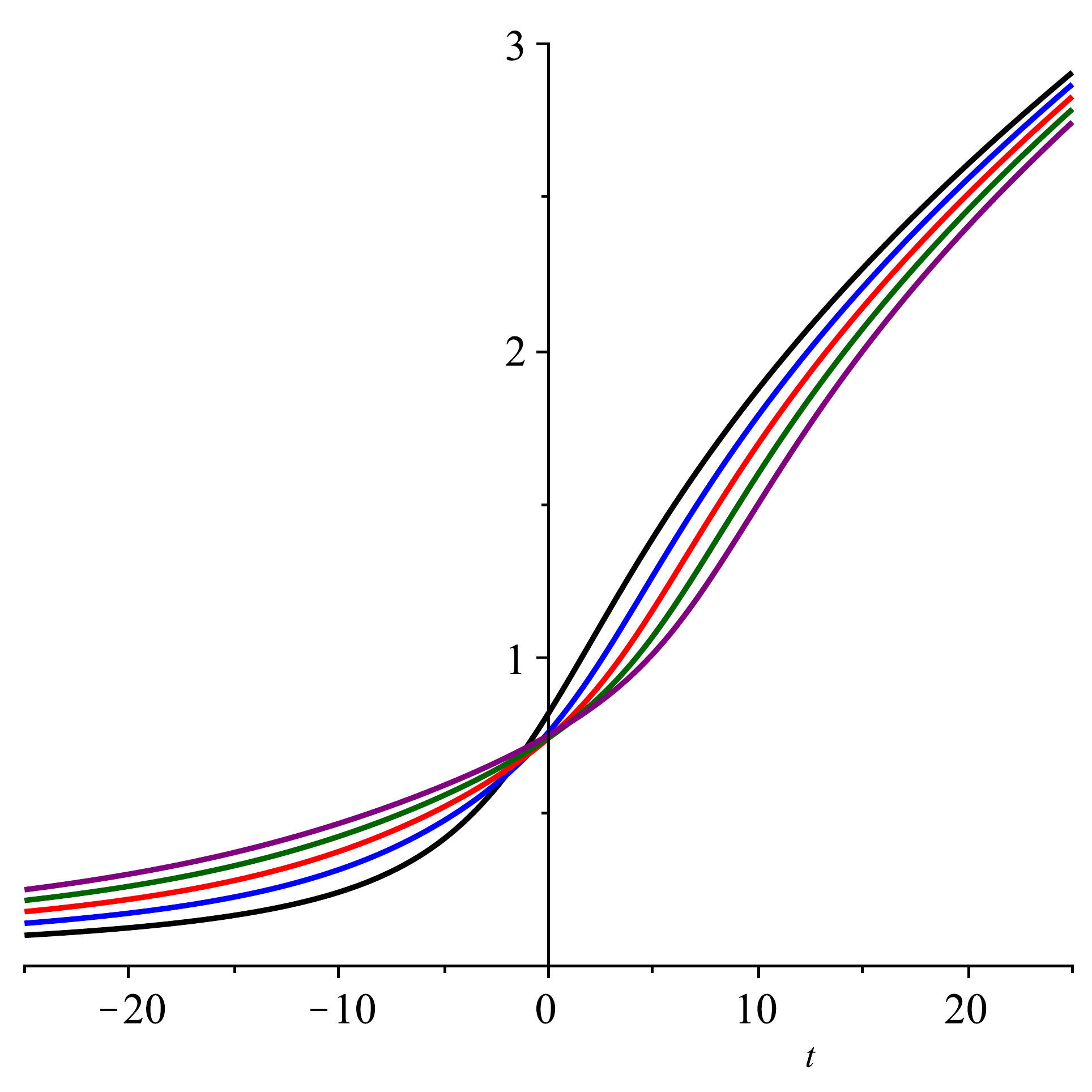} \\
\b_{2n-1}(t;-\tfrac12) & \b_{2n-1}(t;\tfrac12) & \b_{2n-1}(t;\tfrac32)\\
\includegraphics[width=2in]{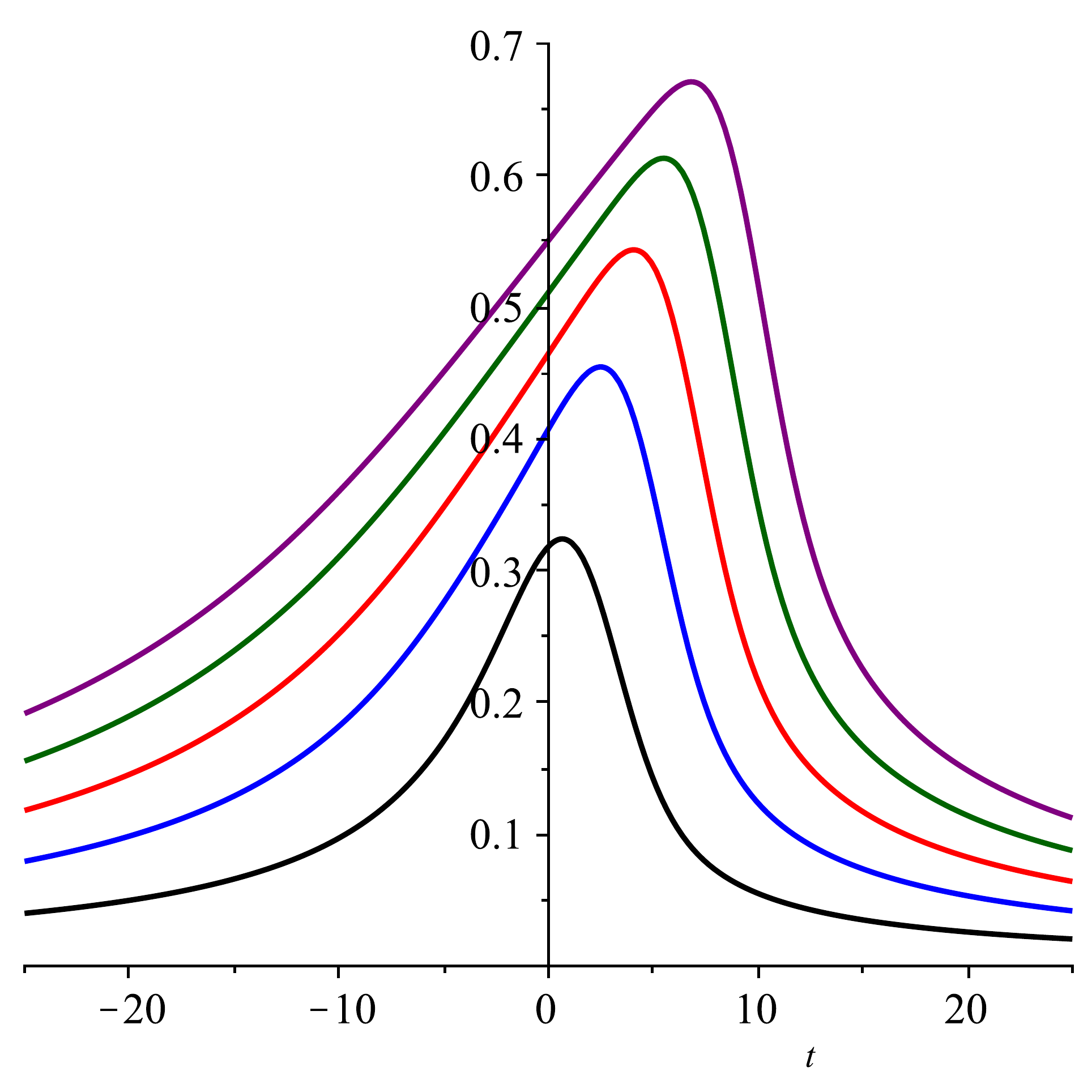} 
& \includegraphics[width=2in]{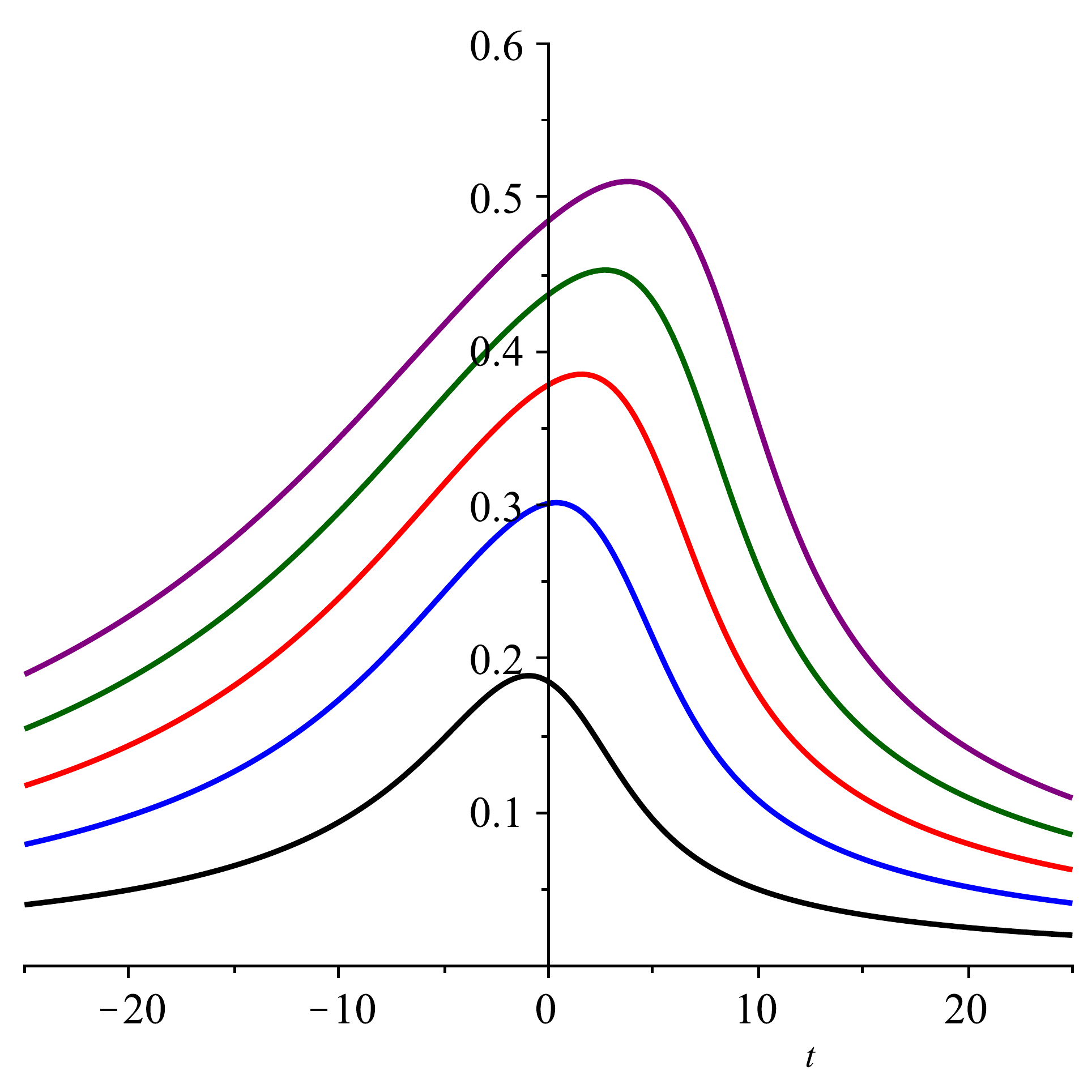} & \includegraphics[width=2in]{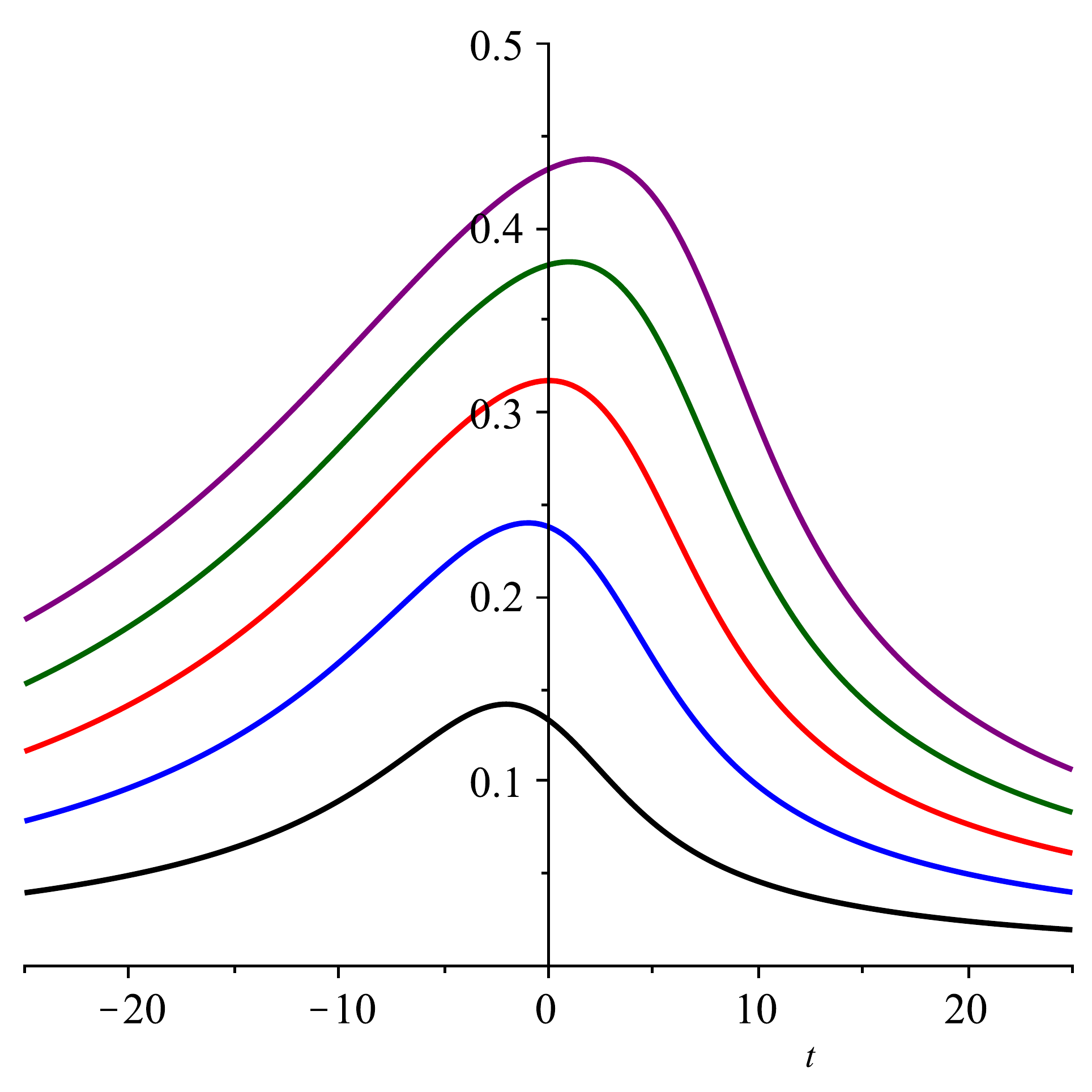} \\
\b_{2n}(t;-\tfrac12) & \b_{2n}(t;\tfrac12) & \b_{2n}(t;\tfrac32)
\end{array}\]
\caption{\label{fig:betan}Plots of the recurrence coefficients $\b_{2n-1}(t;\la)$ and $\b_{2n}(t;\la)$, $n=1,2,\ldots,5$, with $\la=-\tfrac12,\tfrac12,\tfrac32$, for $n=1$ (black), \textcolor{red}{$n=2$ (red)}, \blue{$n=3$ (blue)}, \textcolor{dkg}{$n=4$ (green)} and \textcolor{purple}{$n=5$ (purple)}.}
\end{figure}}

\begin{conjecture}{\rm
\begin{enumerate}\item[]
\item The recurrence coefficient $\b_{2n+1}(t;\la)$ is a monotonically increasing function of $t$.
\item $\b_{2n+2}(t;\la)>\b_{2n}(t;\la)$, for all $t$.
\item The recurrence coefficient $\b_{2n}(t;\la)$ has one maximum at $t=t^*_{2n}$, with $t^*_{2n+2}>t^*_{2n}$.
\end{enumerate}
}\end{conjecture}

\begin{remarks}{\rm 
\begin{enumerate}[(i)] \item[]
\item From the Langmuir lattice \eqref{eq:langlat} we have
\[\frac{1}{\b_{2n+1}}\deriv{\b_{2n+1}}{t}=\b_{2n+2}-\b_{2n},\]
and so $\b_{2n+2}(t;\la)>\b_{2n}(t;\la)$ if and only if $\b_{2n+1}(t;\la)$ is a monotonically increasing function of $t$ since $\b_{2n+1}(t;\la)>0$.
\item Also from the Langmuir lattice we have
\[\frac{1}{\b_{2n}}\deriv{\b_{2n}}{t}=\b_{2n+1}-\b_{2n-1}.\]
and so $\b_{2n}(t;\la)$ has a maximum when $\b_{2n+1}(t;\la)=\b_{2n-1}(t;\la)$. Since $\b_{2n}(t;\la)\to0$ as $t\to\pm\infty$ and $\b_{2n}(t;\la)>0$ then it is a maximum rather than a minimum.
\end{enumerate}
}\end{remarks}
Freud \cite{refFreud76} proved the following result, see also \cite[\S2.3]{refVanAssche07}.
\begin{lemma}\label{lem:Freud} For the weight
 \beq \nonumber w(x) = |x|^{\rho}\exp(-|x|^{6}), \qquad m\in \N,\eeq 
 the recurrence coefficient $\b_n(\rho)$ has the following asymptotic behaviour as $n\to\infty$
\beq \nonumber \lim_{n\to\infty} \frac{\b_n(\rho)}{n^{1/3}} = \frac1{\sqrt[3]{60}}.\eeq
\end{lemma}%
{\begin{corollary}{\label{col:Freud} For the sextic Freud weight \eqref{freud6g}, the recurrence coefficient $\b_n(t;\la)$ has the following asymptotic behaviour as $n\to\infty$
\beq \nonumber \lim_{n\to\infty} \frac{\b_{n}(t;\la)}{n^{1/3}} = \frac1{\sqrt[3]{60}}.\eeq
}\end{corollary}
\begin{proof}
Whilst Lemma \ref{lem:Freud}  applies to the sextic Freud weight \eqref{freud6g} in the case $t = 0$, it is straightforward to extend the proof, for example in \cite[\S2.3]{refVanAssche07}, to the case when $t \not= 0$.
\end{proof}}%
\def\rh{(2\la+1)}
\begin{theorem}{\label{thm:bn}The recurrence coefficient $\b_n(t;\la)$ in the three-term recurrence relation for the sextic Freud weight \eqref{freud6g}
has the {formal} asymptotic expansions, {for fixed $t$ and fixed $\la$},
if $n$ is even 
\begin{subequations}\label{bn:asmyp}\beq
\b_{n}(t;\la) =
 \frac{n^{1/3}}{\k} + \frac{t\k}{90n^{1/3}} - \frac{\rh\k^2}{90n^{2/3}} -\frac{4 t\rh\k}{135n^{4/3}}-\frac{[t^3-135(7\la^2+7\la+2)]\k^2}{36450n^{5/3}} + \O\big(n^{-2}\big),
\eeq
and if $n$ is odd
\beq \b_{n}(t;\la) = 
\frac{n^{1/3}}{\k} +\frac{t\k}{90n^{1/3}}+ \frac{\rh\k^2}{60n^{2/3}}+\frac{7t\rh\k}{270n^{4/3}}-\frac{[2t^3+135(6\la^2+6\la+1)]\k^2}{72900n^{5/3}} + \O\big(n^{-2}\big),
\eeq\end{subequations} with $\k=\sqrt[3]{60}$, 
as $n\to\infty$.}\end{theorem}

\begin{proof}
The recurrence coefficient $\b_{n}$ satisfies the nonlinear discrete equation \eqref{eq:dPI2}, which for $\la\not=-\tfrac12$ has a $(-1)^n$ term which suggests an even-odd dependence in $\b_n(t;\la)$. This dependence needs to be taken into
account to obtain an asymptotic approximation. Therefore we suppose that
\beq \b_n=\begin{cases} u_n,\quad &\text{if}\quad n\quad\text{even},\\
v_n,\quad &\text{if}\quad n\quad\text{odd},\end{cases}\label{dp22:tr}\eeq
where from 
Corollary \ref{col:Freud} 
\[ \lim_{n\to\infty}\frac{u_n}{\sqrt[3]{n}}= \frac{1}{\sqrt[3]{60}},\qquad \lim_{n\to\infty}\frac{v_n}{\sqrt[3]{n}}= \frac{1}{\sqrt[3]{60}},\]
then $(u_n,v_n)$ satisfy
\begin{subequations}\label{eq:dPI2ab}
\begin{align}6u_{n} \big(u_{n+2} v_{n+1} &+ v_{n+1}^2 + 2 v_{n+1} u_{n} + v_{n+1} v_{n-1} + u_{n}^2 + 2 u_{n}v_{n-1} + v_{n-1}^2 + v_{n-1} u_{n-2}\big)\nonumber\\ & 
-2t u_{n}= n,\label{eq:dPI2a}\\
6v_{n} \big(v_{n+2} u_{n+1} &+ u_{n+1}^2 + 2 u_{n+1} v_{n} + u_{n+1} u_{n-1} + v_{n}^2 + 2 v_{n}u_{n-1} + u_{n-1}^2 + u_{n-1} v_{n-2}\big)
\nonumber\\& 
-2t v_{n}= n+2\la+1.\label{eq:dPI2b}
\end{align}\end{subequations}
We remark that the transformation \eqref{dp22:tr} was used by Cresswell and Joshi \cite{refCJ99a} when they derived the continuum limit of \eqref{eq:gendPI2}.
Now suppose that
\begin{subequations}\label{fngn}\begin{align} u_n&= \frac{n^{1/3}}{60^{1/3}} + \sum_{j=0}^5 \frac{a_j}{n^{j/3}} + \O\big(n^{-2}\big),\qquad
v_n= \frac{n^{1/3}}{60^{1/3}} + \sum_{j=0}^5 \frac{b_j}{n^{j/3}} + \O\big(n^{-2}\big),
\end{align}
where $a_j$, $b_j$, $j=0,1,\ldots,5$, are constants to be determined. Then
\begin{align} 
u_{n\pm1}&=
\frac{n^{1/3}}{\k}+a_0+\frac{a_1}{n^{1/3}} + \frac{a_2\k\pm\tfrac13}{\k n^{2/3}} + \frac{a_3}{n} + \frac{a_2\mp\tfrac13a_1}{n^{4/3}} + \frac{(a_5\mp\tfrac23a_2)\k-\tfrac19}{\k n^{5/3}} + \O\big(n^{-2}\big),\\
v_{n\pm1}&= 
\frac{n^{1/3}}{\k}+b_0+\frac{b_1}{n^{1/3}} + \frac{b_2\k\pm\tfrac13}{\k n^{2/3}}+ \frac{b_3}{n} + \frac{b_2\mp\tfrac13b_1}{n^{4/3}} + \frac{(b_5\mp\tfrac23b_2)\k-\tfrac19}{\k n^{5/3}} + \O\big(n^{-2}\big),\\
u_{n\pm2}&=
\frac{n^{1/3}}{\k}+a_0+\frac{a_1}{n^{1/3}} +\frac{a_2\k\pm\tfrac23}{\k n^{2/3}} + \frac{a_3}{n} + \frac{a_2\mp\tfrac23a_1}{n^{4/3}} + \frac{(a_5\mp\tfrac43a_2)\k-\tfrac49}{\k n^{5/3}} + \O\big(n^{-2}\big),\\
v_{n\pm2}&= 
\frac{n^{1/3}}{\k}+b_0+\frac{b_1}{n^{1/3}} + \frac{b_2\k\pm\tfrac23}{\k n^{2/3}} + \frac{b_3}{n} + \frac{b_2\mp\tfrac23b_1}{n^{4/3}} + \frac{(b_5\mp\tfrac43b_2)\k-\tfrac49}{\k n^{5/3}} + \O\big(n^{-2}\big),
\end{align}\end{subequations} with $\k=\sqrt[3]{60}$. Substituting \eqref{fngn} into \eqref{eq:dPI2ab} and equating powers of $n$ gives
\begin{align*} &a_0=b_0=0,\qquad a_1=b_1=\frac{t \k}{90},\qquad a_2=-\frac{\rh\k^2}{90},\quad b_2=\frac{\rh\k^2}{60},\\ &a_3=b_3=0,\qquad
a_4=-\frac{4t\rh\k}{135},\quad b_4=\frac{7t\rh\k}{270},\\
&a_5=-\frac{[t^3-135(7\la^2+7\la+2)]\k^2}{36450},\quad b_5=-\frac{[2t^3+135(6\la^2+6\la+1)]\k^2}{72900},
\end{align*}
with $\k=\sqrt[3]{60}$, and so if $n$ is even then
\begin{subequations}\label{eq:betan_asympt}\beq\b_n=
\ds \frac{n^{1/3}}{\k} + \frac{t\k}{90n^{1/3}} - \frac{\rh\k^2}{90n^{2/3}} -\frac{4 t\rh\k}{135n^{4/3}}-\frac{[t^3-135(7\la^2+7\la+2)]\k^2}{36450n^{5/3}} + \O\big(n^{-2}\big),\eeq whilst if $n$ is odd then
\beq\b_n= \frac{n^{1/3}}{\k} +\frac{t\k}{90n^{1/3}}+ \frac{\rh\k^2}{60n^{2/3}}+\frac{7t\rh\k}{270n^{4/3}}-\frac{[2t^3+135(6\la^2+6\la+1)]\k^2}{72900n^{5/3}} + \O\big(n^{-2}\big),\eeq\end{subequations}
as required.
\end{proof}

Plots of $\b_n(t;\tfrac12)$, for $n=1,2,\ldots,100$, with $t=0,1,2,3,5,10$ are given in Figure~\ref{fig:betala12} and 
plots of $\b_n(2;\la)$, for $n=1,2,\ldots,100$, with $\la=0,\tfrac12,1,2,3,5$ are given in Figure~\ref{fig:betat2}. In these plots, the blue dots $\blue{\bullet}$ are $\b_n(t;\la)$ for $n$ even and the red dots $\textcolor{red}{\bullet}$ are $\b_n(t;\la)$ for $n$ odd. 
The solid lines are the asymptotics \eqref{eq:betan_asympt} and the dashed line is
\beq\frac{n^{1/3}}{\k} + \frac{t\k}{90n^{1/3}},\label{eq:betan_asympt2}\eeq
with $\k=\sqrt[3]{60}$.
\def\FreudFig#1{\includegraphics[width=2in]{#1}}
{\begin{figure}[ht]
\[\begin{array}{c@{\quad}c@{\quad}c}
\FreudFig{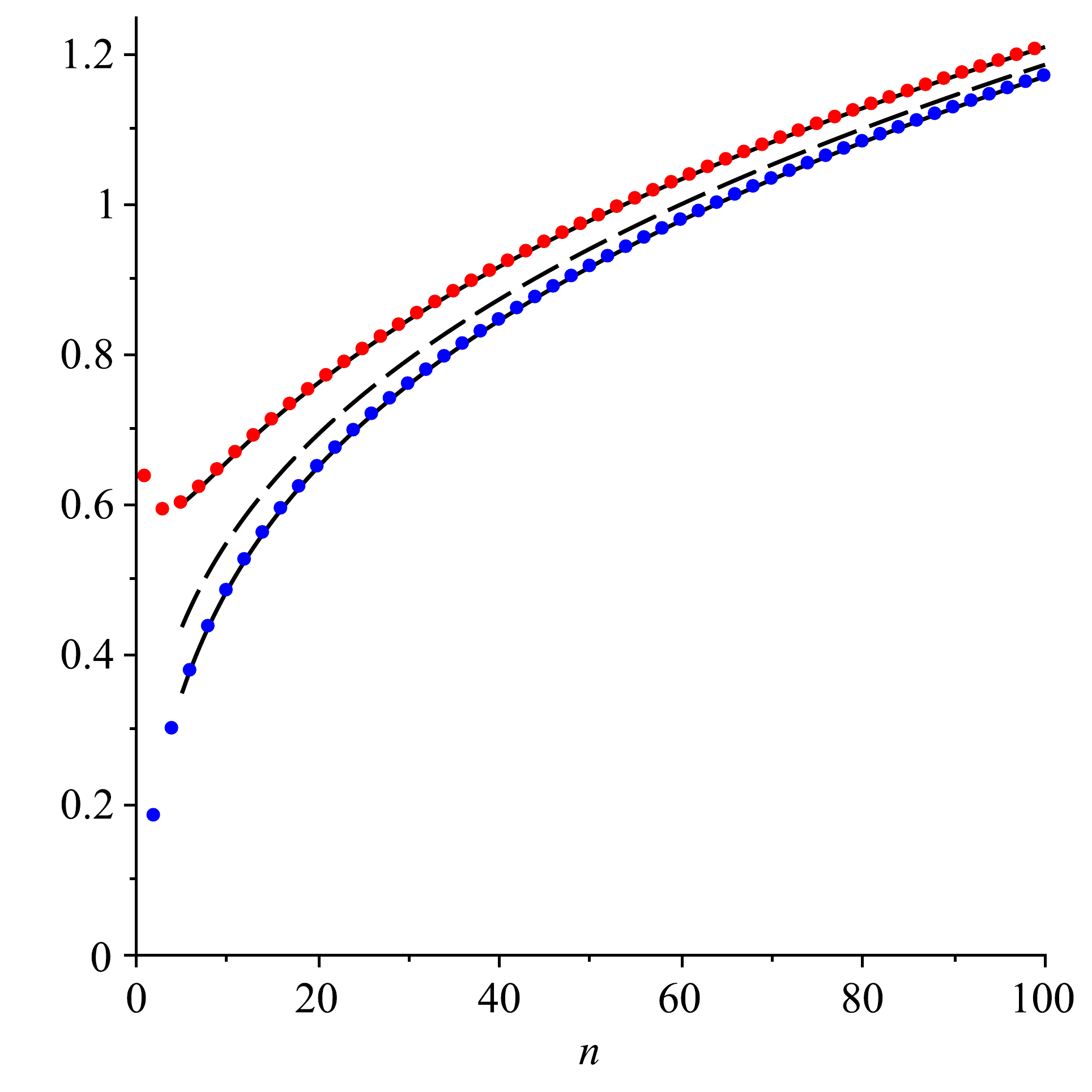} & \FreudFig{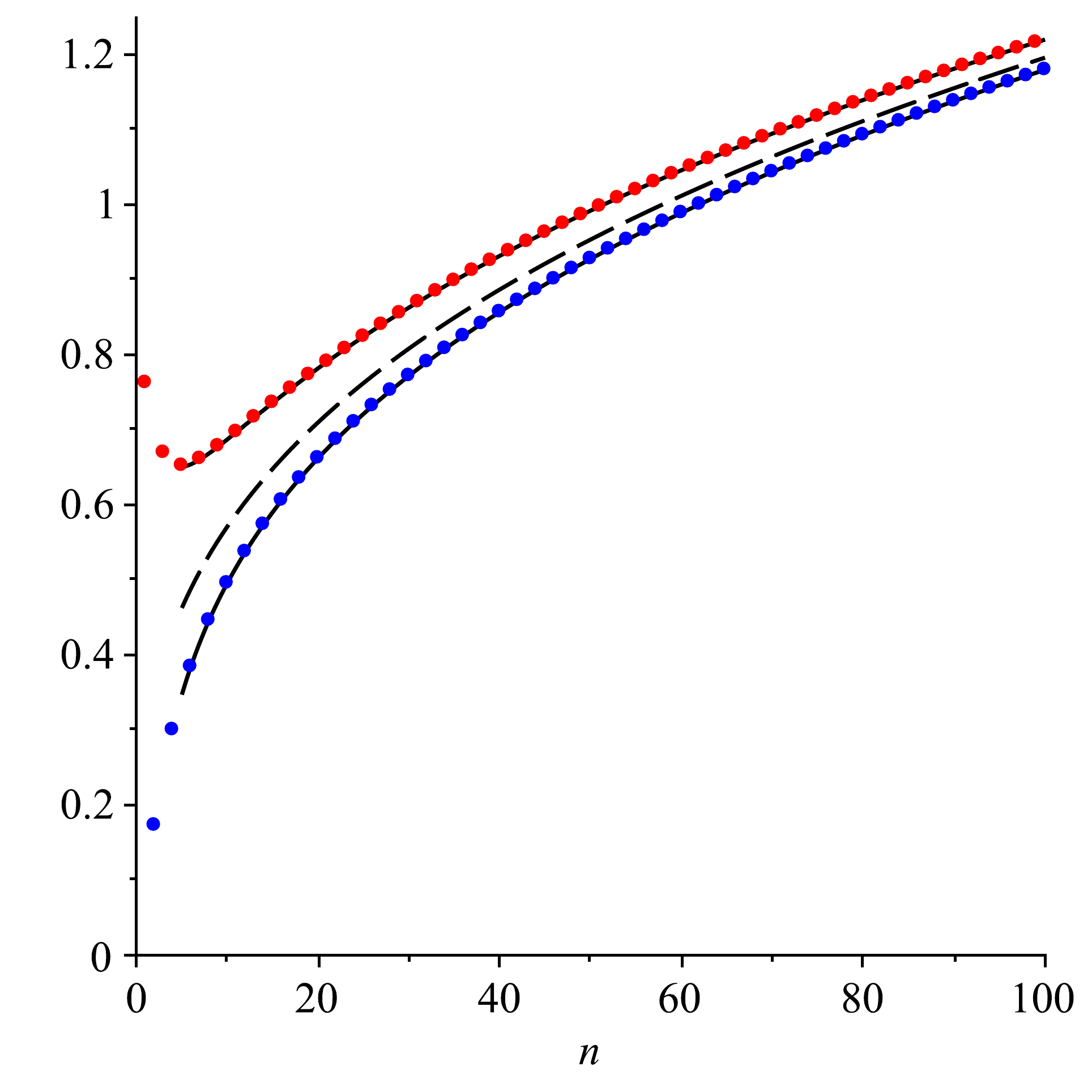} & \FreudFig{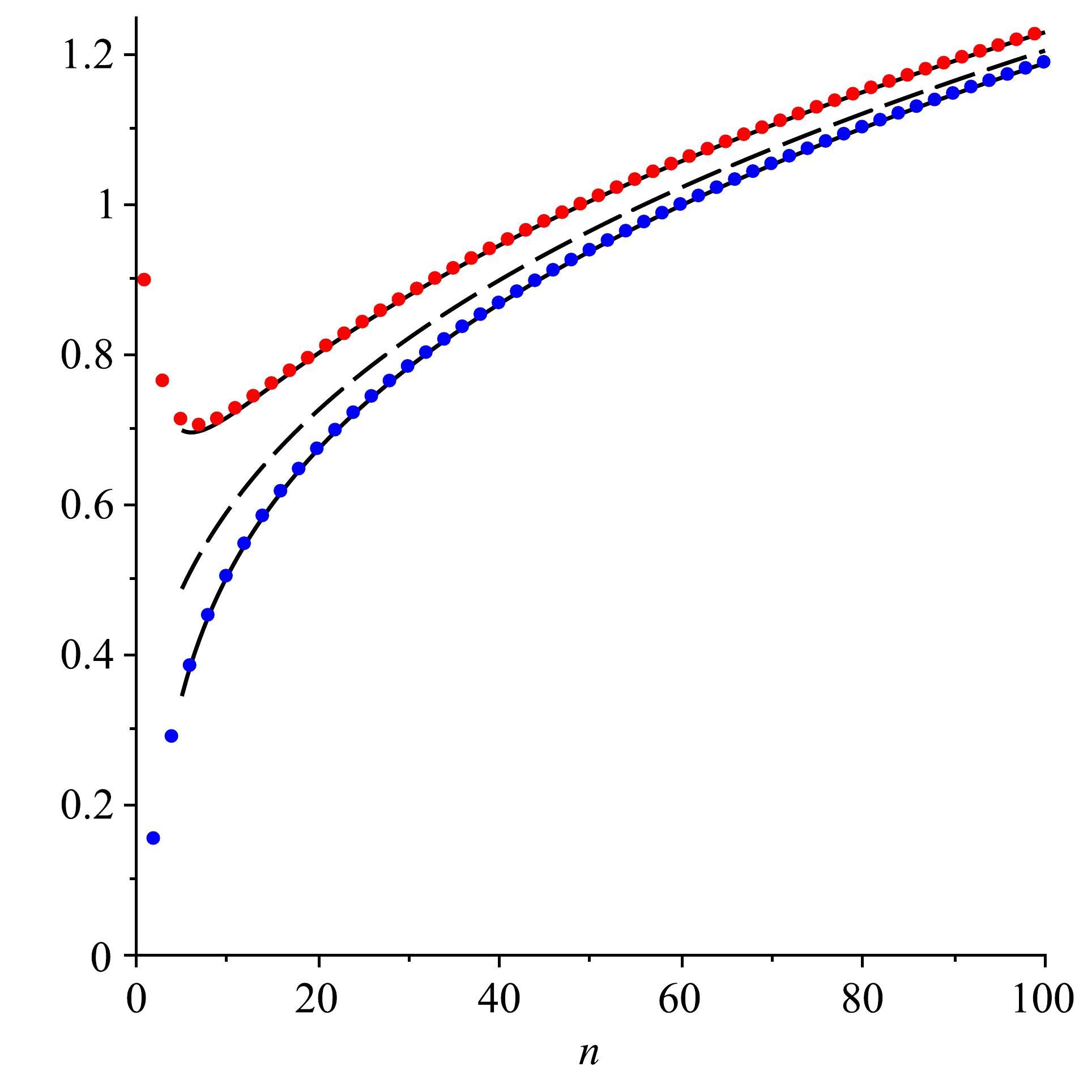}\\
\b_n(0;\tfrac12) &\b_n(1;\tfrac12) &\b_n(2;\tfrac12) \\
\FreudFig{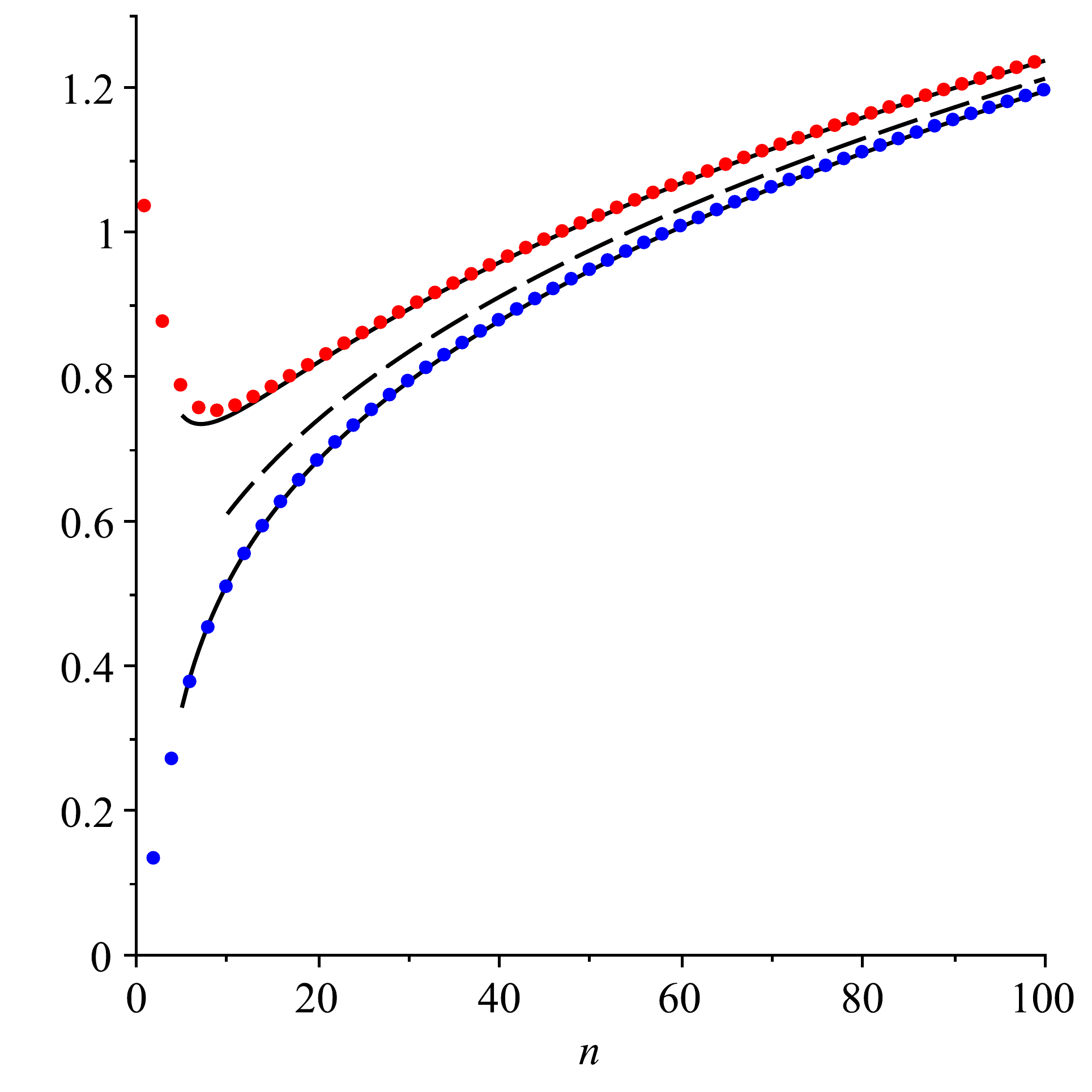} &\FreudFig{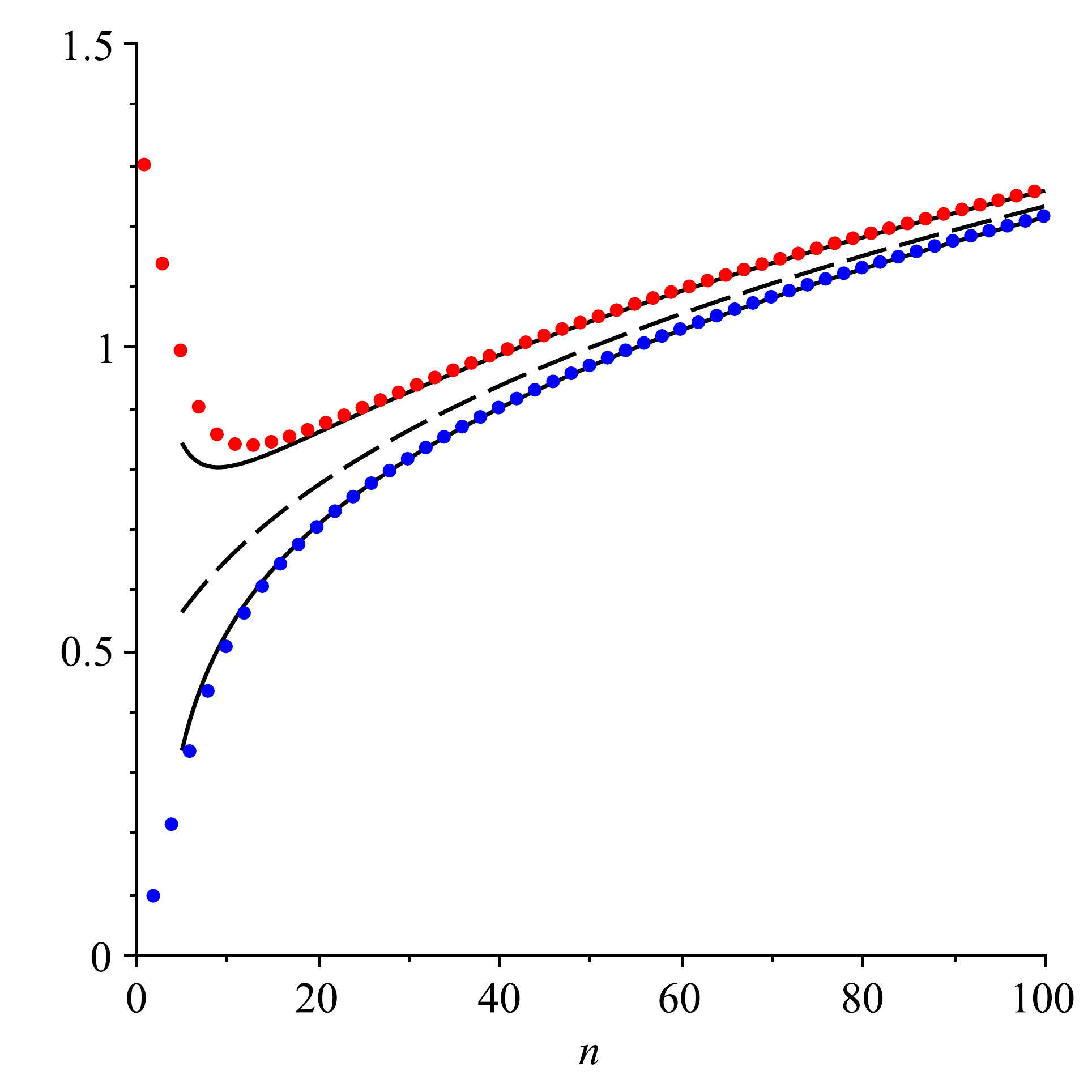} &\FreudFig{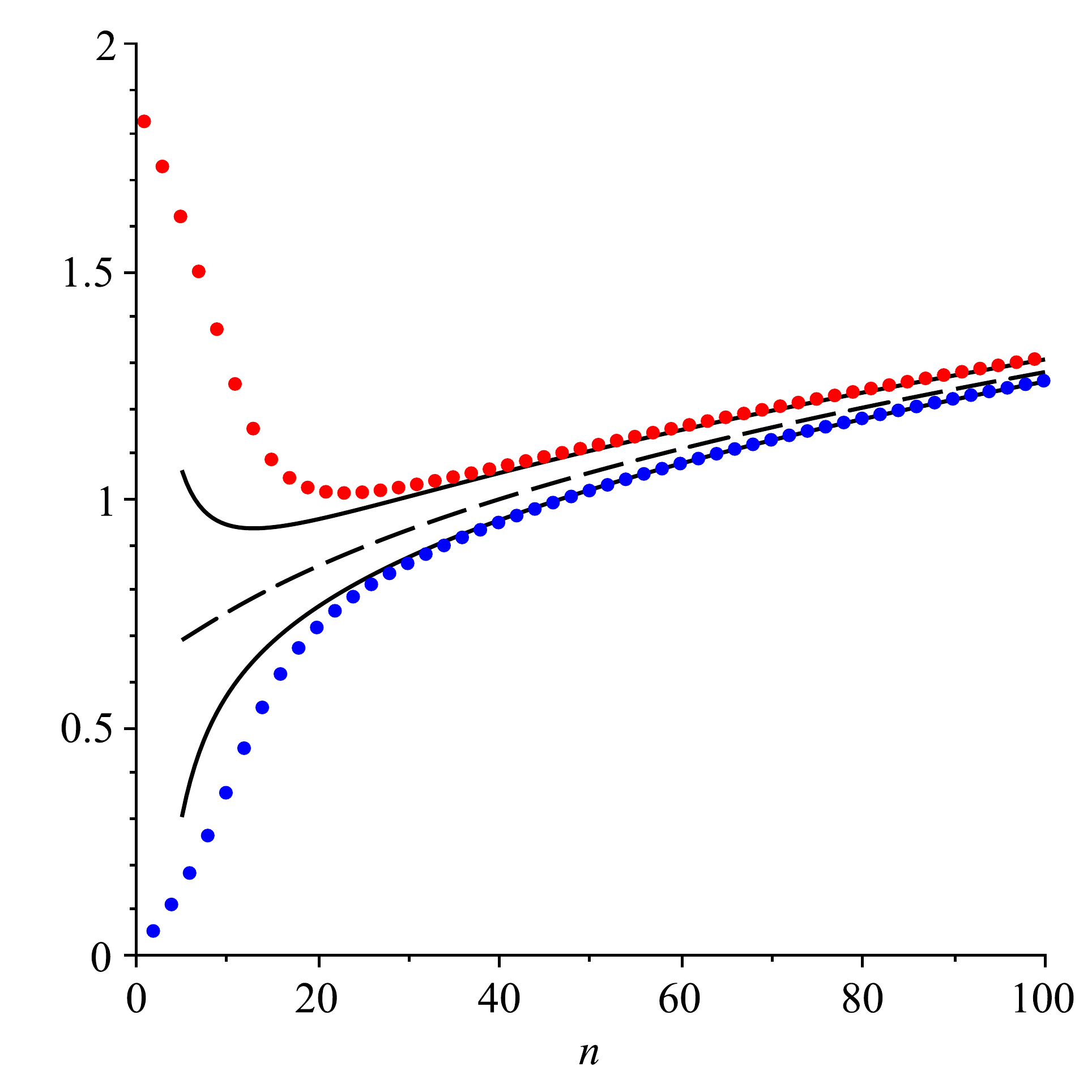} \\
\b_n(3;\tfrac12) &\b_n(4;\tfrac12) &\b_n(10;\tfrac12) 
\end{array}\]\caption{\label{fig:betala12}Plots of $\b_n(t;\tfrac12)$, for $n=1,2,\ldots,100$, with $t=0,1,2,3,5,10$. The blue dots $\blue{\bullet}$ are $\b_n(t;\tfrac12)$ for $n$ even and the red dots $\textcolor{red}{\bullet}$ are $\b_n(t;\tfrac12)$ for $n$ odd. 
The solid lines are the asymptotics \eqref{eq:betan_asympt} and the dashed line is \eqref{eq:betan_asympt2}.
}
\end{figure}

\begin{figure}[ht]
\[\begin{array}{ccc}
\FreudFig{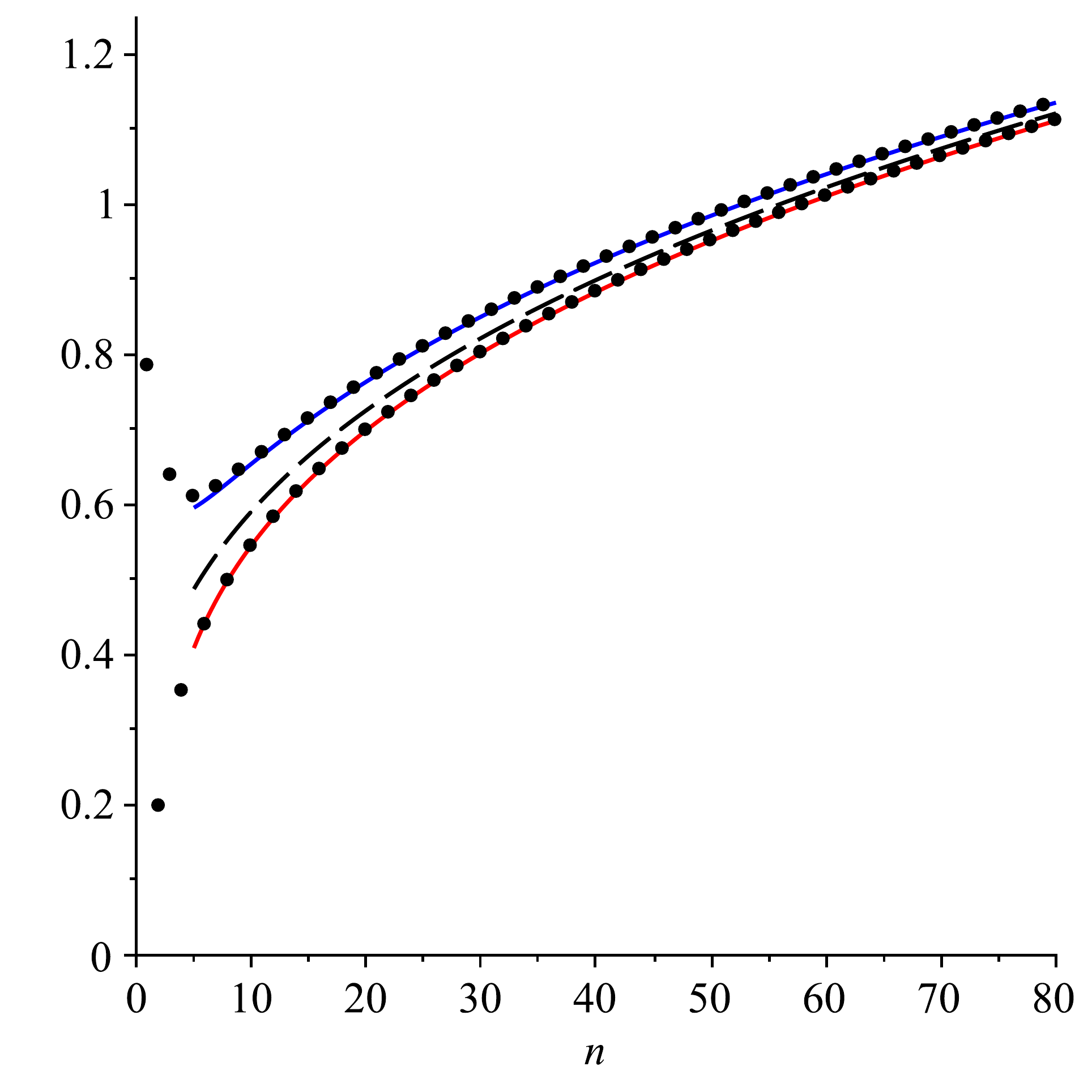} & \FreudFig{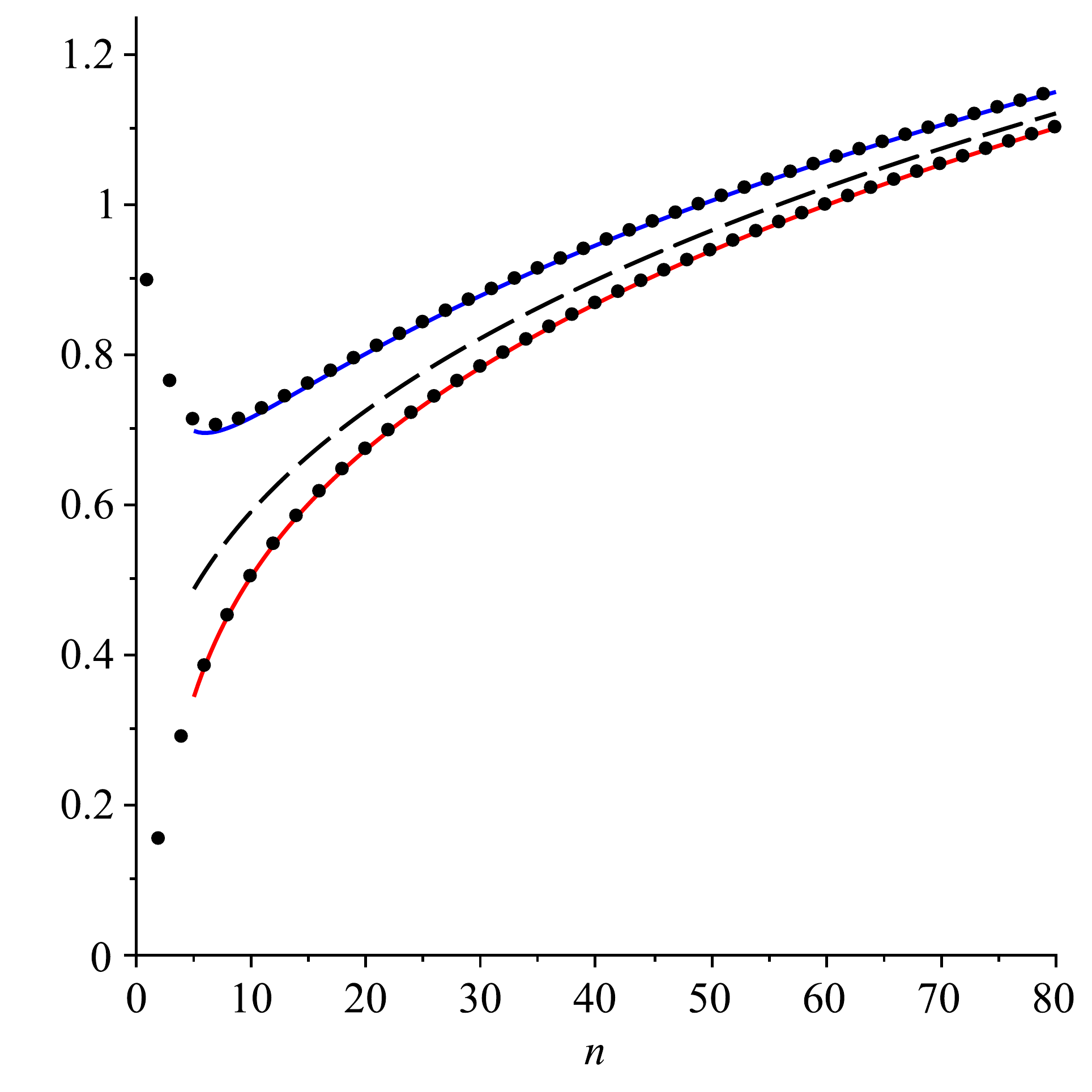}& \FreudFig{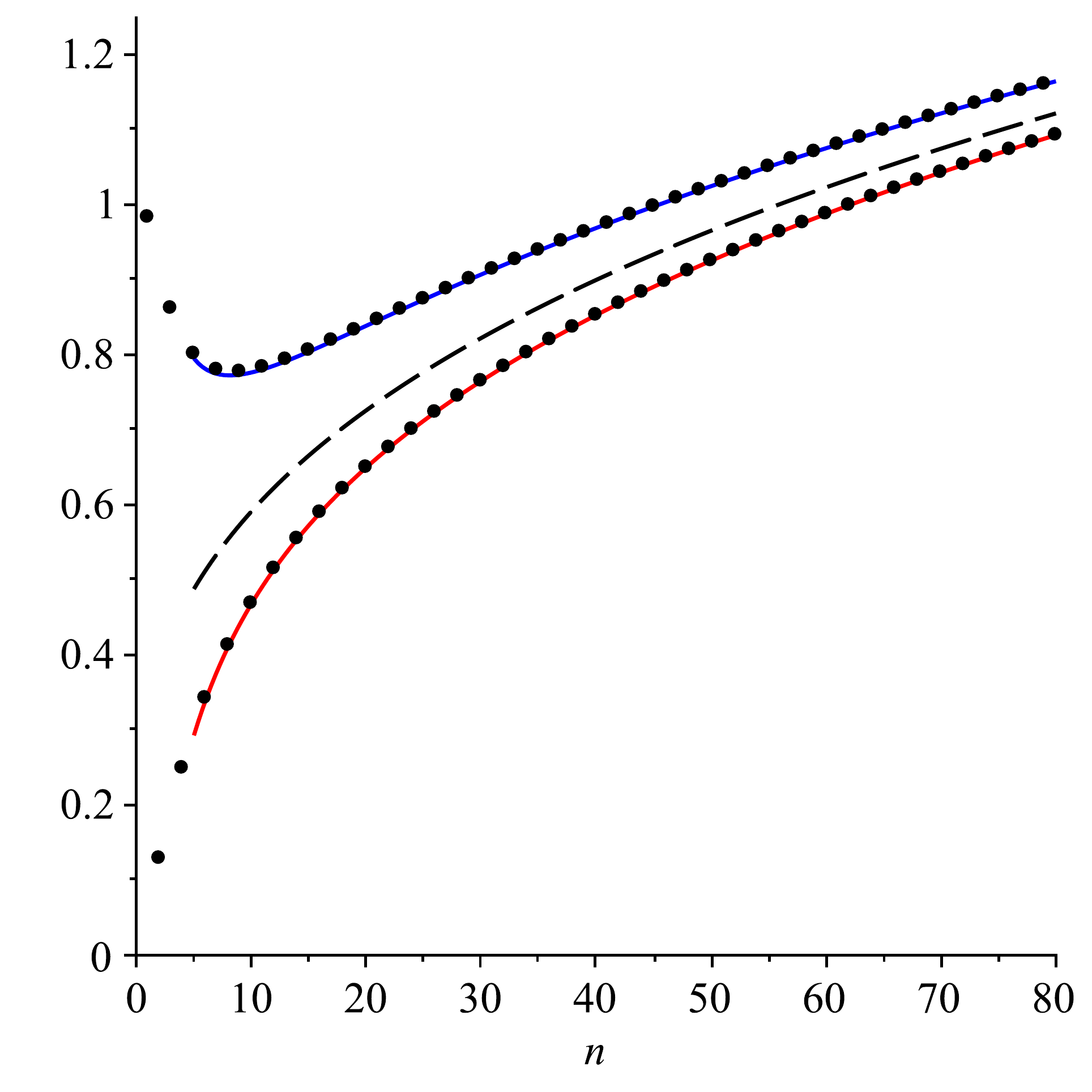} \\
\b_n(2;0) &\b_n(2;\tfrac12) &\b_n(2;1)\\
\FreudFig{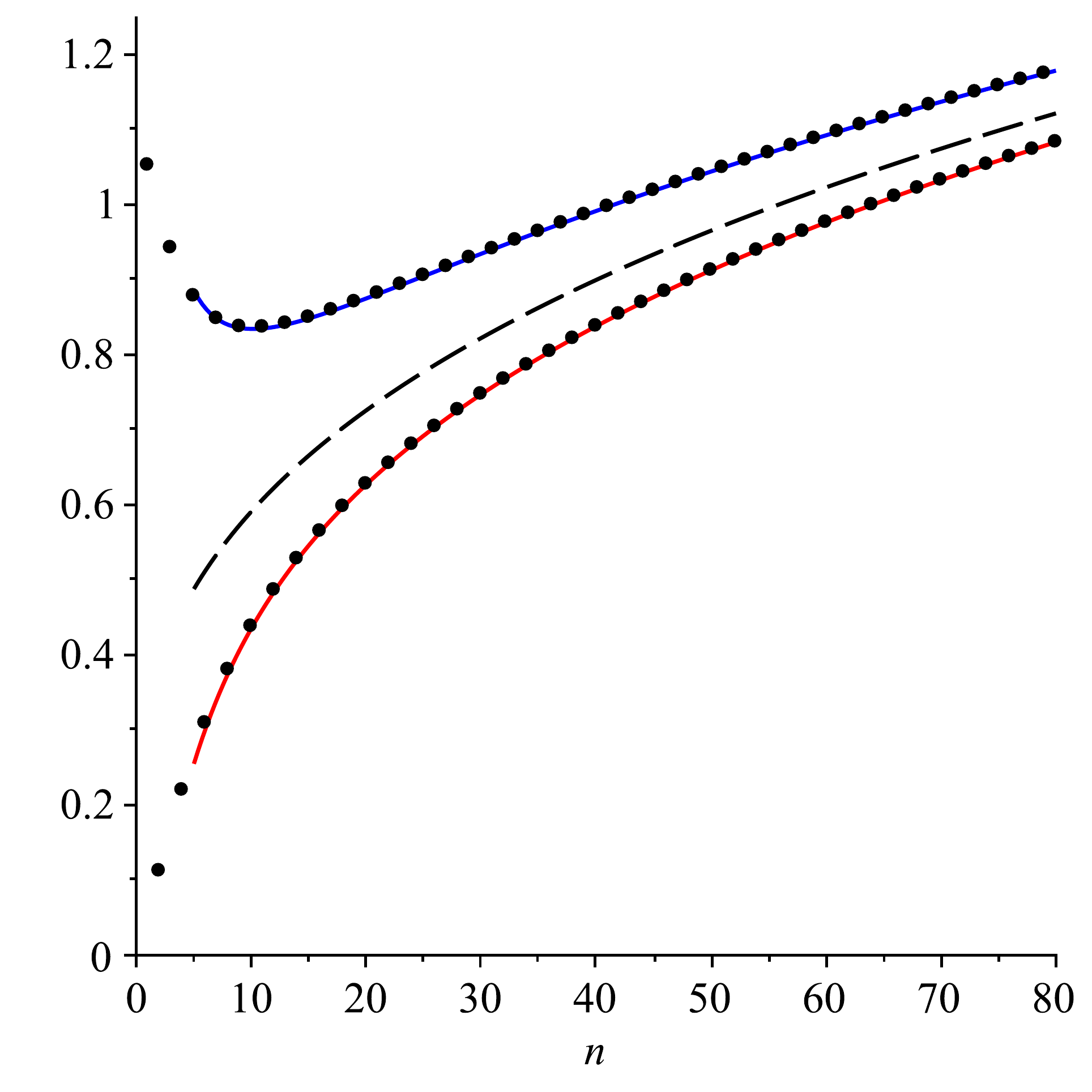} & \FreudFig{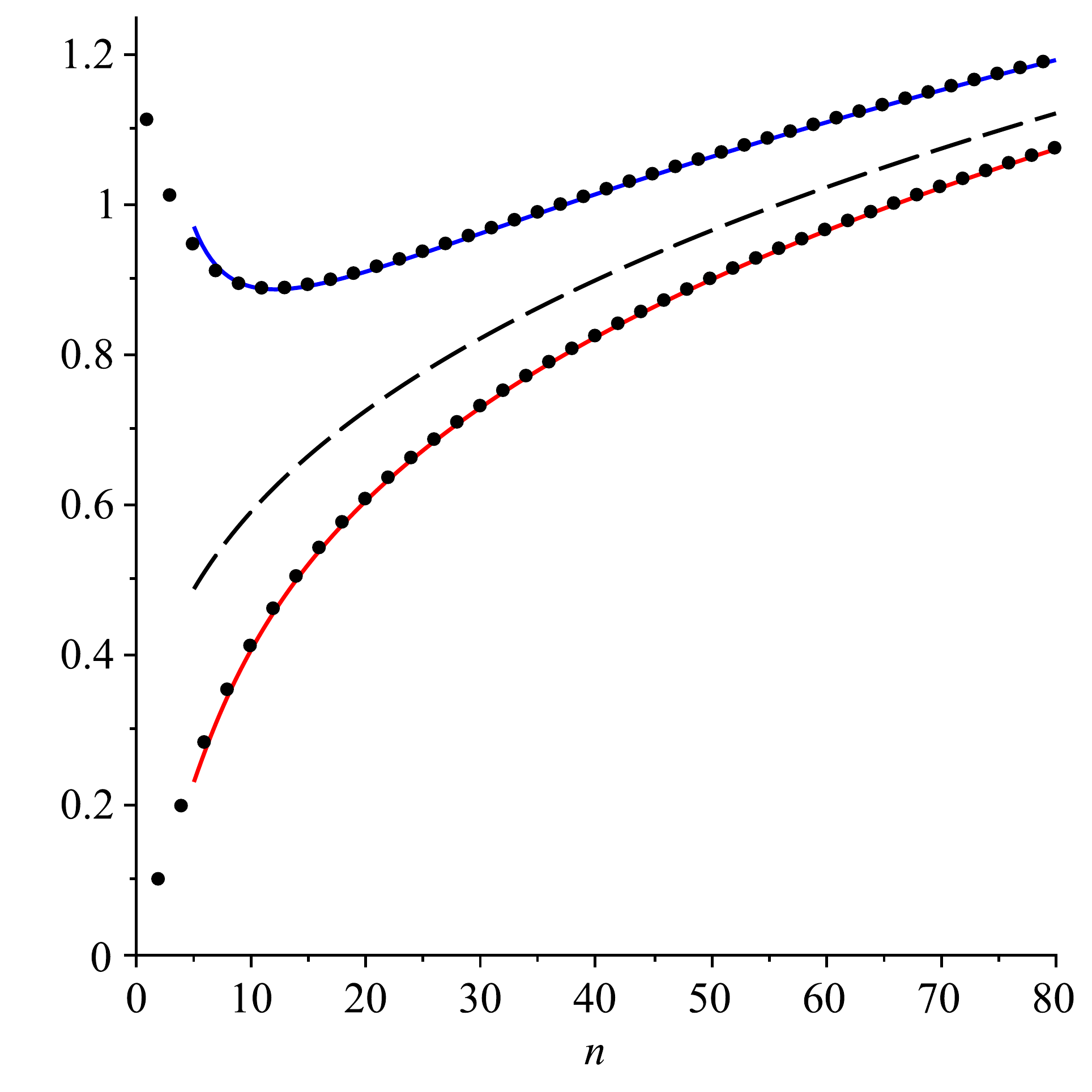} & \FreudFig{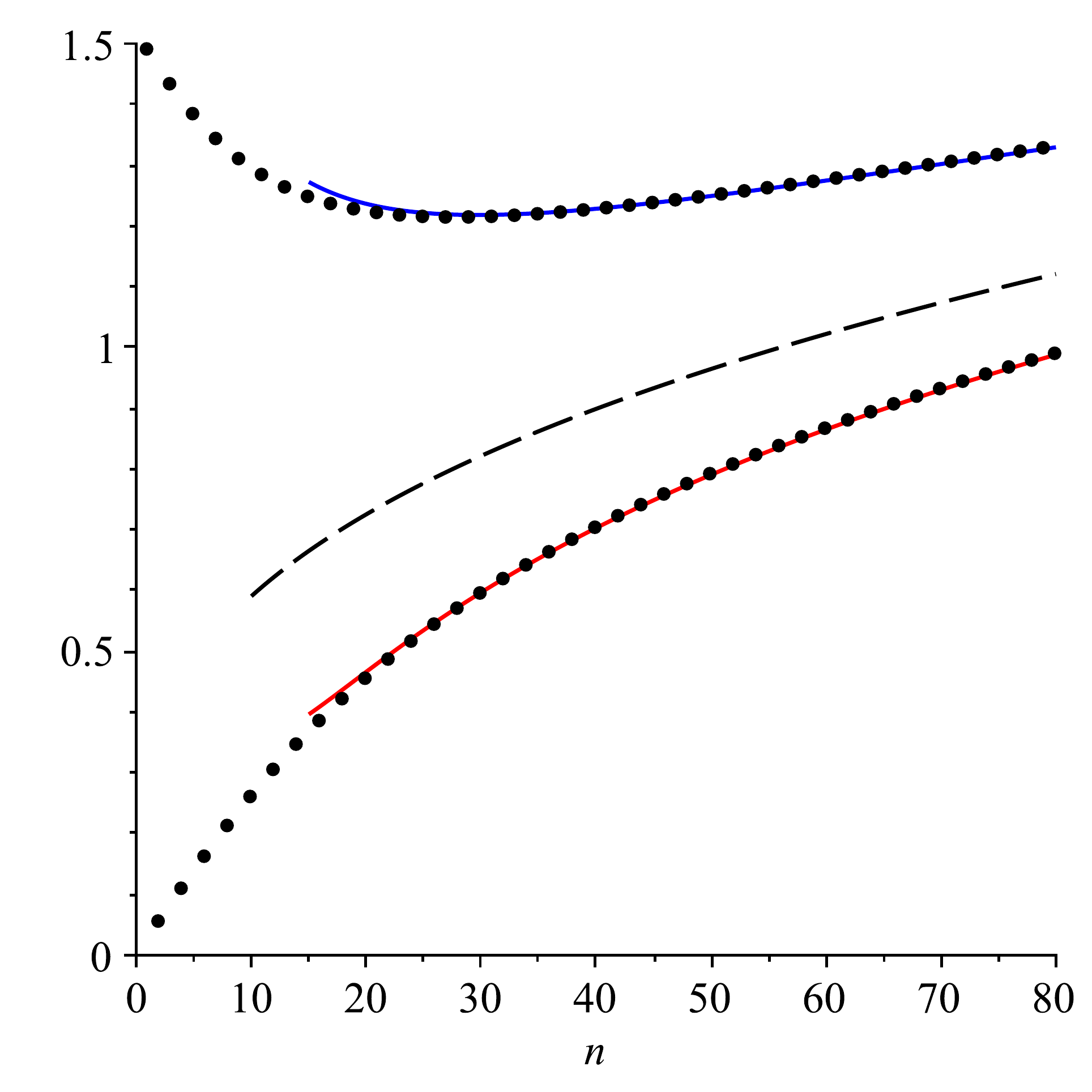}\\
\b_n(2;2) &\b_n(2;3) &\b_n(2;5) 
\end{array}\]
\caption{\label{fig:betat2}Plots of $\b_n(2;\la)$, for $n=1,2,\ldots,100$, with $\la=0,\tfrac12,1,2,3,5$. The blue dots $\blue{\bullet}$ are $\b_n(2;\la)$ for $n$ even and the red dots $\textcolor{red}{\bullet}$ are $\b_n(2;\la)$ for $n$ odd. 
The solid lines are the asymptotics \eqref{eq:betan_asympt} and the dashed line is \eqref{eq:betan_asympt2}.}
\end{figure}

\begin{remark}{\rm In \cite{refWZC}, Wang, Zhu and Chen state that 
\[ \b_n(t;\la) \sim\frac{2^{4/3}t}{\Theta_n(t;\la)} + \frac{\Theta_n(t;\la)}{45\times 2^{7/3}},\qquad\text{as}\quad n\to\infty,\]
where
\[\Theta_n^3(t;\la)= 48600\left[ 2n+2\la+1+\sqrt{(2n+2\la+1)^2 - \frac{32}{405}t^3}\right].\]
From this it can be shown that as $n\to\infty$
\[ \b_n(t;\la)= \frac{n^{1/3}}{\k} +\frac{t\k}{90n^{1/3}}+\frac{\rh\k^2}{360n^{2/3}}-\frac{t\rh\k}{540n^{4/3}}+\O\big(n^{-5/3}\big),\]
with $\k=\sqrt[3]{60}$, though this is not given in \cite{refWZC},
which is the average of the asymptotic expressions for $\b_n(t;\la)$ for $n$ even and odd given by \eqref{bn:asmyp}. 
}\end{remark}

\section{Higher Freud weights}\label{sec:Gen46810Freud}
In this section we discuss generalised higher Freud weights of the form
\begin{equation} \w(x;t)=|x|^{2\la+1}\exp\left(-x^{2m}+tx^2\right),\qquad \la>-1,\qquad x,t\in\R,\label{hfreud}\end{equation}
in the cases when $m=4$ and $m=5$ and show {how some of the results} for the generalised sextic Freud weight \eqref{freud6g} can be extended to these higher weights.

In Lemma \ref{lem:Freud6weight} we showed that, when $m=3$ in \eqref{hfreud}, the first moment of the generalised sextic Freud weight is given by 
\[ \begin{split}
\mu_0(t;\la)& =\int_{-\infty}^{\infty} |x|^{2\la+1}\exp(-x^6+tx^2)\,\d x = \int_0^\infty s^{\la}\exp(ts-s^3)\,\d s \nonumber\\
& = \tfrac13\Gamma(\tfrac13\la+\tfrac13) \;\HyperpFq12(\tfrac13\la+\tfrac13;\tfrac13,\tfrac23;(\tfrac13t)^3) 
+ \tfrac13 \,t\,\Gamma(\tfrac13\la+\tfrac23) \;\HyperpFq12(\tfrac13\la+\tfrac23;\tfrac23,\tfrac43;(\tfrac13t)^3)\nonumber\\ &\qquad\qquad
+ \tfrac16\,t^2\,\Gamma(\tfrac13\la+1) \;\HyperpFq12(\tfrac13\la+1;\tfrac43,\tfrac53;(\tfrac13t)^3), 
\end{split}\]
where 
$\HyperpFq12(a_1;b_1,b_2;z)$ is the generalised hypergeometric function and satisfies the third-order equation
\[ \deriv[3]{\vph}{t}-\tfrac13t\deriv{\vph}{t}-\tfrac13(\la+1)\vph=0.\]%
\comment{Also, the first recurrence coefficient $\ds\b_1(t;\la)=\deriv{}{t}\mu_0(t;\la)$ satisfies the second-order equation
\[\deriv[2]{\b_1}{t}+3\b_1\deriv{\b_1}{t}+\b_1^3-\tfrac13t\b_1=\tfrac13(\la+1).\]}%
 Recall that when $m=2$ in \eqref{hfreud}, 
the first moment of the generalised quartic Freud weight
\[ \w(x;t)=|x|^{2\la+1}\exp\left(-x^4+tx^2\right),\qquad \la>-1,\qquad x,t\in\R,\] is given by (cf. \cite{refCJK})
 \begin{align}
\mu_0(t;\la)&=\int_{-\infty}^{\infty} |x|^{2\la+1}\exp(-x^4+tx^2)\,\d x = \frac{\Gamma(\la)}{2^{(\la+1)/2}} \exp(\tfrac18 t^2) D_{-\la-1}(-\tfrac12\sqrt{2}\,t) \nonumber\\
&= \tfrac12 \Gamma(\tfrac12\la+\tfrac12) \;\HyperpFq11(\tfrac12\la+\tfrac12;\tfrac12;\tfrac1{4}t^2)+ \tfrac12t \Gamma(\tfrac12\la+1) \;\HyperpFq11(\tfrac12\la+1;\tfrac32;\tfrac1{4}t^2),
\label{mu0:Freud4}
\end{align}
where $\HyperpFq11(a;b;z)$ is the confluent hypergeometric function, which is equivalent to the Kummer function $M(a,b,z)$. The relationship between the parabolic cylinder function $D_{\nu}(\zeta)$ and the Kummer function $M(a,b,z)$ is given in \cite[\S13.6]{refNIST}.
Further $\mu_0(t;\la)$ given by \eqref{mu0:Freud4} satisfies the second-order equation
\[ \deriv[2]{\vph}{t}-\tfrac12t\deriv{\vph}{t}-\tfrac12(\la+1)\vph=0.\]
\comment{and $\ds\b_1(t;\la)=\deriv{}{t}\mu_0(t;\la)$ satisfies the Riccati equation
\[\deriv{\b_1}{t}+\b_1^2-\tfrac12t\b_1=\tfrac12(\la+1).\]}%
\subsection{The generalised octic Freud weight}
\begin{lemma}{\label{lemma71}For the generalised octic Freud weight
\beq \w(x;t)=|x|^{2\la+1}\exp\left(-x^8+tx^2\right),\qquad \la>-1,\qquad x,t\in\R,\label{Freud8}\eeq
then the first moment is given by 
\[ \begin{split}
\mu_0(t;\la)& =\int_{-\infty}^{\infty} |x|^{2\la+1}\exp(-x^8+tx^2)\,\d x = \int_0^\infty s^{\la}\exp(-s^4+ts)\,\d s \\
& = \tfrac14\Gamma(\tfrac14\la+\tfrac14) \;\HyperpFq13\left(\tfrac14\la+\tfrac14;\tfrac14,\tfrac12,\tfrac34;(\tfrac1{4}t)^4\right) + \tfrac14 \,t
\,\Gamma(\tfrac14\la+\tfrac12) \;\HyperpFq13\left(\tfrac14\la+\tfrac12;\tfrac12,\tfrac34,\tfrac54;(\tfrac1{4}t)^4\right) \\ &\qquad
+ \tfrac18\,t^2\,\Gamma(\tfrac14\la+\tfrac34) \;\HyperpFq13\left(\tfrac14\la+\tfrac34;\tfrac34,\tfrac54,\tfrac32;(\tfrac1{4}t)^4\right)
+ \tfrac1{24}\,t^3\,\Gamma(\tfrac14\la+1) \;\HyperpFq13\left(\tfrac14\la+1;\tfrac54,\tfrac32,\tfrac74;(\tfrac1{4}t)^4\right),
\end{split}\]
where 
$\HyperpFq13(a_1;b_1,b_2,b_3;z)$ is the generalised hypergeometric function.
\comment{The general solution of the fourth-order equation
\[ \deriv[4]{\vph}{t}-\tfrac14t\deriv{\vph}{t}-\tfrac14(\la+1)\vph=0,\]
is given by
\[\begin{split} \vph(t)&=c_1 \;\HyperpFq13(\tfrac14\la+\tfrac14;\tfrac14,\tfrac12,\tfrac34;(\tfrac1{4}t)^4) + c_2t\;\HyperpFq13(\tfrac14\la+\tfrac12;\tfrac12,\tfrac34,\tfrac54;(\tfrac1{4}t)^4) \\ &\qquad
+ c_3t^2 \;\HyperpFq13(\tfrac14\la+\tfrac34;\tfrac43,\tfrac53;(\tfrac1{4}t)^4)+c_4\,t^3 \;\HyperpFq13(\tfrac14\la+1;\tfrac54,\tfrac32,\tfrac74;(\tfrac1{4}t)^4),\end{split}\]
with $c_1$, $c_2$, $c_3$ and $c_4$ constants.}%
Further $\mu_0(t;\la)$ satisfies the fourth-order equation
\[ \deriv[4]{\vph}{t}-\tfrac14t\deriv{\vph}{t}-\tfrac14(\la+1)\vph=0,\]%
{and the first recurrence coefficient $\ds\b_1(t;\la)=\deriv{}{t}\ln\mu_0(t;\la)$ satisfies the third-order equation
\beq\deriv[3]{\b_1}{t}+4\b_1\deriv[2]{\b_1}{t}+3\left(\deriv{\b_1}{t}\right)^2+6\b_1^2\deriv{\b_1}{t}+\b_1^4-\tfrac14t\b_1=\tfrac14(\la+1).
\label{eq:beta81}\eeq}}\end{lemma}
\begin{proof} The proof is analogous that for the generalised sextic Freud weight \eqref{freud6g} in Lemma \ref{lem:Freud6weight}.
\end{proof}%

{\begin{lemma} \label{lem72}The recurrence coefficient $\b_n(t;\la)$ for the generalised octic Freud weight \eqref{Freud8} has the asymptotics as $t\to\infty$
\begin{align*} \b_{2n}(t;\la) &=  \frac{n}{3\,t}+\frac{2^{5/3}\,n(3n-1-2\la)}{9\,t^{7/3}}+\O(t^{-11/3}), \\
\b_{2n+1}(t;\la) &= \ds(\tfrac14t)^{1/3}-\frac{3n-\la+1}{3\,t}-\frac{2^{2/3}\,[36n^2-12n(4\la-1)+6\la^2-18\la+7]}{27\,t^{7/3}}+\O(t^{-11/3}), \end{align*}
 and as  $t\to-\infty$
\begin{align*} \b_{2n}(t;\la) &= -\frac{n}{t}+\O(t^{-5}),\qquad
\b_{2n+1}(t;\la)= -\frac{n+\la+1}{t}+\O(t^{-5}).
\end{align*}
\end{lemma}}%
{\begin{proof}
Using Laplace's method it follows that as $t\to\infty$
\[\begin{split} 
\mu_0(t;\la) &= \int_0^\infty s^{\la}\exp(ts-s^4)\,\d s = t^{(\la+1)/3}\int_0^\infty \xi^{\la}\exp\{t^{4/3}\xi(1-\xi^3)\}\,\d\xi \\ 
&= (\tfrac1{6}\pi)^{1/2}\,(\tfrac14 t)^{(\la-1)/3} \exp\left\{3(\tfrac14t)^{4/3}\right\}\left[1+\O(t^{-4/3})\right] \\
\mu_2(t;\la) &= \mu_0(t;\la+1)= (\tfrac1{6}\pi)^{1/2}\,(\tfrac14 t)^{\la/3} \exp\left\{3(\tfrac14t)^{4/3}\right\}\left[1+\O(t^{-4/3})\right]
\end{split}\]
and so
\[ \b_1(t;\la) =\frac{\mu_2(t;\la)}{\mu_0(t;\la)}= (\tfrac14t)^{1/3}\left[1+\O(t^{-4/3})\right],\qquad\text{as}\quad t\to\infty.\]
Hence we suppose that as $t\to\infty$
\[ \b_1(t;\la) = (\tfrac14t)^{1/3} + \frac{a_1}{t}+\frac{a_2}{t^{7/3}}+ \O(t^{-11/3}).\]
Substituting this into \eqref{eq:beta81} and equating coefficients of powers of $t$ gives 
\[ a_1= \tfrac13(\la-1),\qquad a_2 = -{2^{2/3}(6\la^2 - 18\la + 7)}/{27},\]
and so
\[ \b_1(t;\la) = \tfrac13\sqrt{3t} + \frac{\la-1}{3\,t}-\frac{2^{2/3}(6\la^2 - 18\la + 7)}{27\,t^{7/3}}+ \O(t^{-11/3}).\]
For the analogous result for the generalised sextic weight \eqref{freud6g}, in 
Lemma \ref{betaasymp} we were able to use the system \eqref{sys:ubn} to derive the leading asymptotics for $\b_n(t;\la)$ as $t\to\pm\infty$. However for the generalised octic Freud weight \eqref{Freud8} we don't have the analog of  \eqref{sys:ubn}. Instead we can use the Langmuir lattice \eqref{eq:langlat}.
Now suppose that $u_n(t;\la)=\b_{2n}(t;\la)$ and $v_n(t;\la)=\b_{2n+1}(t;\la)$,  then from the Langmuir lattice \eqref{eq:langlat} we obtain
\beq u_{n+1}= u_n+\deriv{}{t}\ln v_n,\qquad v_{n+1}= v_n+\deriv{}{t}\ln u_{n+1}, \label{eq:langlat8}\eeq
with
\[ u_0=0,\qquad v_0(t;\la)=\b_1(t;\la)= \tfrac13\sqrt{3t} + \frac{\la-1}{3\,t}-\frac{2^{2/3}(6\la^2 - 18\la + 7)}{27\,t^{7/3}}+ \O(t^{-11/3}).\]
Next we assume that as $t\to\infty$
\begin{align} u_n= \frac{a_{n,1}}{t} + \frac{a_{n,2}}{t^{7/3}} + \O(t^{11/3}),\qquad
v_n = \tfrac13\sqrt{3t} + \frac{b_{n,1}}{t} + \frac{b_{n,2}}{t^{7/3}} + \O(t^{11/3}).
\end{align}
Substituting this into \eqref{eq:langlat8} and equating powers of $t$ gives the recurrence relations
\[ \begin{split}
a_{n+1,1}&=a_{n,1}+\tfrac13,\qquad
b_{n+1,1}=b_{n,1}-1,
\end{split}\qquad
\begin{split}
a_{n+1,2}&=a_{n,2}-\tfrac43 2^{2/3}\,b_{n,1},\qquad
b_{n+1,2}=b_{n,2} - \frac{4a_{n+1,2}}{3a_{n+1,1}}.
\end{split}\]
Solving these with initial conditions 
\[ a_{0,1}=a_{0,2}=0,\qquad b_{0,1}= \tfrac13(\la-1),\qquad b_{0,2}=-\frac{2^{2/3}(6\la^2 - 18\la + 7)}{27},\]
gives
\begin{align*} a_{n,1}&=\tfrac13n, && a_{n,2}=\frac{2^{5/3}\,n(3n-1-2\la)}{9},\\
b_{n,1}&=\tfrac13(\la-1-3n),&& b_{n,2}=-\frac{2^{2/3}\,[36n^2-12n(4\la-1)+6\la^2-18\la+7]}{27},
\end{align*}
as required.
Using Watson's Lemma it follows that as $t\to-\infty$
\[ \mu_0(t;\la) = {\Gamma(\la+1)}{(-t)^{-\la-1}}\left[1+\O(t^{-4})\right],\qquad \mu_2(t;\la) = {\Gamma(\la+2)}{(-t)^{-\la-2}}\left[1+\O(t^{-4})\right],\]
and so
\[ \b_1(t;\la) =-\frac{\la+1}{t}  + \O(t^{-5}),\qquad\text{as}\quad t\to-\infty.\]
Consequently, we now assume that 
\[ u_n=  \frac{c_{n}}{t} + \O(t^{-5}),\qquad v_n=  \frac{d_{n}}{t} + \O(t^{-5}),\qquad\text{as}\quad t\to-\infty.\]
Substituting these into \eqref{eq:langlat8} and equating powers of $t$ gives the recurrence relations
\[c_{n+1}=c_{n}-1,\qquad d_{n+1}=d_{n}-1. \]
Solving these with initial conditions
\[ c_{0}=0,\qquad d_0= -(\la+1),\]
gives
\[ c_n=-n,\qquad d_n=-(\la+n+1),\]
as required.
\end{proof}}

\subsection{The generalised decic Freud weight}
\begin{lemma}{\label{lemma73}For the generalised decic Freud weight
\[ \w(x;t)=|x|^{2\la+1}\exp\left(-x^{10}+tx^2\right),\qquad \la>-1,\qquad x,t\in\R,\label{Freud10}\]
then the first moment is given by 
\[ \begin{split}
\mu_0(t;\la)& =\int_{-\infty}^{\infty} |x|^{2\la+1}\exp(-x^{10}+tx^2)\,\d x = \int_0^\infty s^{\la}\exp(-s^5+ts)\,\d s \\
& = \tfrac15\Gamma(\tfrac15\la+\tfrac15) \;\HyperpFq14\left(\tfrac15\la+\tfrac15;\tfrac15,\tfrac25,\tfrac35,\tfrac45;(\tfrac15t)^5\right) 
+ \tfrac15 \,t
\,\Gamma(\tfrac15\la+\tfrac25) \;\HyperpFq14\left(\tfrac15\la+\tfrac25;\tfrac25,\tfrac35,\tfrac45,\tfrac65;(\tfrac15t)^5\right) \\ &\qquad
+ \tfrac1{10}\,t^2\,\Gamma(\tfrac15\la+\tfrac35) \;\HyperpFq14\left(\tfrac15\la+\tfrac35;\tfrac35,\tfrac45,\tfrac65,\tfrac75;(\tfrac15t)^5\right) 
+ \tfrac1{30}\,t^3\,\Gamma(\tfrac15\la+\tfrac45) \;\HyperpFq14\left(\tfrac15\la+\tfrac45;\tfrac45,\tfrac65,\tfrac75,\tfrac85;(\tfrac15t)^5\right)\\&\qquad
+ \tfrac1{120}\,t^4\,\Gamma(\tfrac15\la+1) \;\HyperpFq14\left(\tfrac15\la+1;\tfrac65,\tfrac75,\tfrac85,\tfrac95;(\tfrac15t)^5\right)
\end{split}\]
where 
$\HyperpFq14(a_1;b_1,b_2,b_3,b_{4};z)$ is the generalised hypergeometric function. 
 Further $\mu_0(t;\la)$ satisfies the fifth-order equation
\[ \deriv[5]{\vph}{t} -\tfrac15{t} \deriv{\vph}{t} - \tfrac15(\la+1)\,\vph=0,\]
{and the first recurrence coefficient $\ds\b_1(t;\la)=\deriv{}{t}\ln\mu_0(t;\la)$ satisfies the fourth-order equation
\[\deriv[4]{\b_1}{t}+5\b_1\deriv[3]{\b_1}{t}+10\left(\deriv{\b_1}{t}+\b_1^2\right)\deriv[2]{\b_1}{t}+15\b_1\left(\deriv{\b_1}{t}\right)^2+10\b_1^3\deriv{\b_1}{t}+\b_1^5-\tfrac15t\b_1=\tfrac15(\la+1).\]}}\end{lemma}
\begin{proof} As for Lemma \ref{lemma71} above, the proof is analogous to that for the generalised sextic Freud weight \eqref{freud6g} in Lemma \ref{lem:Freud6weight}.
\end{proof}
{\begin{lemma}The recurrence coefficient $\b_n(t;\la)$ for the generalised decic Freud weight \eqref{Freud10} has the asymptotics as $t\to\infty$
\begin{align*} \b_{2n}(t;\la) &= \frac{n}{4\,t} +\frac{5^{5/4}\,n(4n-1-2\la)}{32\,t^{9/4}}+\O(t^{-7/2}), \\
\b_{2n+1}(t;\la) &= \ds(\tfrac15t)^{1/4}-\frac{8n-2\la+3}{8\,t}-\frac{5^{5/4}\,[40n^2-20(2\la-1)+4\la^2-16\la+9]}{128\,t^{9/4}}+\O(t^{-7/2}),
 \end{align*}
 and as  $t\to-\infty$
\begin{align*} \b_{2n}(t;\la) &= -\frac{n}{t}+\O(t^{-6}),\qquad
\b_{2n+1}(t;\la)= -\frac{n+\la+1}{t}+\O(t^{-6}).
\end{align*}
\end{lemma}
\begin{proof}
The proof is very similar to the proof of Lemma \ref{lem72} above. In this case using Laplace's method 
\[\mu_0(t;\la)=(\tfrac1{10}\pi)^{1/2}\,(\tfrac15 t)^{(2\la-3)/8} \exp\left\{4(\tfrac15t)^{5/4}\right\}\left[1+\O(t^{-5/4})\right], \qquad  \text{as}\quad t\to\infty\]
and from Watson's lemma 
\[ \mu_0(t;\la) = {\Gamma(\la+1)}{(-t)^{-\la-1}}\left[1+\O(t^{-5})\right],\qquad  \text{as}\quad t\to-\infty.\]
\end{proof}}

\section{Discussion}
In this paper we studied generalised sextic Freud weights, the associated orthogonal polynomials and the recurrence coefficients. We also investigated the interesting structural connections between the moments of the weight when the order of the polynomial in the exponential factor of the weight is increased. Further analysis of this interesting class of generalised higher order Freud polynomials and their properties, such as asymptotic expressions for the polynomials and their greatest zeros, is currently in progress. It is important to note that our technique of expressing the Hankel determinants $\Delta_{2n}$ and $\Delta_{2n+1}$, for symmetric weights such as the generalised Freud weights, in terms of smaller Hankel determinants $\A_n$ and $\B_n$, as was done in \S\ref{sec:op}, had several benefits. The method resulted in expressions for $\b_{2n}$ and $\b_{2n+1}$ in terms of $\A_n$ that allowed the derivation of nonlinear discrete and nonlinear differential equations for $\b_n$ which do not appear to exist when using $\Delta_n$. Futhermore, from a computational point of view, the numerical evaluation of $\b_n$ seems to be much faster when using the determinants $\A_n$ and $\B_n$ rather than $\Delta_n$.
\section*{Acknowledgements} 
We gratefully acknowledge the support of a Royal Society Newton Advanced Fellowship NAF$\backslash$R2$\backslash$180669.
PAC thanks Evelyne Hubert, Arieh Iserles, Ana Loureiro and Walter Van Assche
for their helpful comments and illuminating discussions. 
PAC would like to thank the Isaac Newton Institute for Mathematical Sciences for support and hospitality during the programme ``Complex analysis: techniques, applications and computations" when the work on this paper was undertaken. This work was supported by EPSRC grant number EP/R014604/1.
We also thank the referees for helpful suggestions and additional references.

\def\ams{American Mathematical Society}
\def\AAM{Acta Appl. Math.}
\def\ARMA{Arch. Rat. Mech. Anal.}
\def\bull{Acad. Roy. Belg. Bull. Cl. Sc. (5)}
\def\AC{Acta Crystrallogr.}
\def\AM{Acta Metall.}
\def\ampa{Ann. Mat. Pura Appl. (IV)}
\def\AP{Ann. Phys., Lpz.}
\def\APNY{Ann. Phys., NY}
\def\APP{Ann. Phys., Paris}
\def\BAMS{Bull. Amer. Math. Soc.}
\def\CJP{Can. J. Phys.}
\def\cmp{Commun. Math. Phys.}
\def\CMP{Commun. Math. Phys.}
\def\cpam{Commun. Pure Appl. Math.}
\def\CPAM{Commun. Pure Appl. Math.}
\def\CQG{Classical Quantum Grav.}
\def\crp{C.R. Acad. Sc. Paris}
\def\CSF{Chaos, Solitons \&\ Fractals}
\def\DE{Diff. Eqns.}
\def\DU{Diff. Urav.}
\def\ejam{Europ. J. Appl. Math.}
\def\EJAM{Europ. J. Appl. Math.}
\def\funk{Funkcial. Ekvac.}
\def\FUNK{Funkcial. Ekvac.}
\def\IP{Inverse Problems}
\def\JAMS{J. Amer. Math. Soc.}
\def\JAP{J. Appl. Phys.}
\def\JCP{J. Chem. Phys.}
\def\JDE{J. Differ. Equ.}
\def\JFM{J. Fluid Mech.}
\def\JJAP{Japan J. Appl. Phys.}
\def\JP{J. Physique}
\def\JPhCh{J. Phys. Chem.}
\def\JMAA{J. Math. Anal. Appl.}
\def\JMMM{J. Magn. Magn. Mater.}
\def\JMP{J. Math. Phys.}
\def\jmp{J. Math. Phys}
\def\JNMP{J. Nonl. Math. Phys.}
\def\jpa{J. Phys. A}
\def\JPA{J. Phys. A}
\def\JPB{J. Phys. B: At. Mol. Phys.} 
\def\jpb{J. Phys. B: At. Mol. Opt. Phys.} 
\def\JPC{J. Phys. C: Solid State Phys.} 
\def\JPCM{J. Phys: Condensed Matter} 
\def\JPD{J. Phys. D: Appl. Phys.}
\def\JPE{J. Phys. E: Sci. Instrum.}
\def\JPF{J. Phys. F: Metal Phys.}
\def\JPG{J. Phys. G: Nucl. Phys.} 
\def\jpg{J. Phys. G: Nucl. Part. Phys.} 
\def\JSP{J. Stat. Phys.}
\def\JOSA{J. Opt. Soc. Am.}
\def\JPSJ{J. Phys. Soc. Japan}
\def\JQSRT{J. Quant. Spectrosc. Radiat. Transfer}
\def\LMP{Lett. Math. Phys.}
\def\LNC{Lett. Nuovo Cim.}
\def\NC{Nuovo Cim.}
\def\NIM{Nucl. Instrum. Methods}
\def\NL{Nonlinearity}
\def\NMJ{Nagoya Math. J.}
\def\NP{Nucl. Phys.}
\def\pl{Phys. Lett.}
\def\PL{Phys. Lett.}
\def\PMB{Phys. Med. Biol.}
\def\PR{Phys. Rev.}
\def\PRL{Phys. Rev. Lett.}
\def\PRS{Proc. R. Soc.}
\def\prsl{Proc. R. Soc. Lond. A}
\def\PRSL{Proc. R. Soc. Lond. A}
\def\PS{Phys. Scr.}
\def\PSS{Phys. Status Solidi}
\def\PTRS{Phil. Trans. R. Soc.}
\def\RMP{Rev. Mod. Phys.}
\def\RPP{Rep. Prog. Phys.}
\def\RSI{Rev. Sci. Instrum.}
\def\SAM{Stud. Appl. Math.}
\def\sam{Stud. Appl. Math.}
\def\SSC{Solid State Commun.}
\def\SST{Semicond. Sci. Technol.}
\def\SUST{Supercond. Sci. Technol.}
\def\ZP{Z. Phys.}
\def\JCAM{J. Comput. Appl. Math.}

\def\OUP{Oxford University Press}
\def\CUP{Cambridge University Press}
\def\AMS{American Mathematical Society}

\def\bibitm{\vspace{-7.75pt}\bibitem}

\def\refjl#1#2#3#4#5#6#7{\bibitm{#1}\textrm{\frenchspacing#2}, \textrm{#3},
\textit{\frenchspacing#4}, \textbf{#5}\ (#7)\ #6.}

\def\refjnl#1#2#3#4#5#6#7{\bibitm{#1}{\frenchspacing\rm#2}, #3, 
{\frenchspacing\it#4}, {\bf#5} (#6) #7.}

\def\refpp#1#2#3#4{\bibitm{#1} \textrm{\frenchspacing#2}, \textrm{#3}, #4.}

\def\refjltoap#1#2#3#4#5#6#7{\bibitm{#1} \textrm{\frenchspacing#2}, \textrm{#3},
\textit{\frenchspacing#4} (#7). 
#6.}

\def\refbk#1#2#3#4#5{\bibitm{#1} \textrm{\frenchspacing#2}, \textit{#3}, #4, #5.}

\def\refcf#1#2#3#4#5#6#7{\bibitm{#1} \textrm{\frenchspacing#2}, \textrm{#3},
in: \textit{#4}, {\frenchspacing#5}, pp.\ #7, #6.}

\end{document}